\documentclass[letter,11pt]{article}
\usepackage[letterpaper, margin=1in]{geometry}
\usepackage{fullpage}
\usepackage{todonotes}
\usepackage{datetime}
\usepackage[T1]{fontenc}
\usepackage[utf8]{inputenc}
\usepackage{pgfplots}
\usepackage{amsmath, amssymb}
\usepackage{amsthm}
\usepackage{color}
\usepackage[numbers]{natbib}
\usepackage{cases}
\usepackage{xparse}
\usepackage{appendix}
\usepackage{algorithm}
\usepackage{algpseudocode}
\usepackage{verbatim, xspace}
\usepackage{tikz}
\usetikzlibrary{arrows}
\usetikzlibrary{calc}
\usetikzlibrary{decorations.pathreplacing,angles,quotes}

\newtheorem{theorem}{Theorem}
\newtheorem{definition}[theorem]{Definition}
\newtheorem{lemma}[theorem]{Lemma}
\newtheorem{corollary}[theorem]{Corollary}
\newtheorem{observation}[theorem]{Corollary}
\newtheorem{conjecture}[theorem]{Conjecture}
\newtheorem{fact}[theorem]{Fact}

\newtheorem{problem}{Open Question}

\numberwithin{theorem}{section}

\newcommand{\floor}[1]{\left\lfloor #1 \right\rfloor}
\newcommand{\ceil}[1]{\left\lceil #1 \right\rceil}

\newcommand{\eps}{\varepsilon}
\newcommand{\Oh}{\mathcal{O}}
\newcommand{\Os}{\mathcal{O}^{*}}
\newcommand{\Otilde}{\widetilde{\mathcal{O}}}
\newcommand{\Ot}{\Otilde}

\newcommand{\nat}{\mathbb{N}}
\newcommand{\pow}{\mathrm{pow}}
\newcommand{\poly}{\mathrm{poly}}
\newcommand{\polylog}{\mathrm{polylog}}

\DeclareMathOperator*{\argmin}{arg\,min} 
\DeclareMathOperator*{\argmax}{arg\,max}

\newcommand{\zzlarge}{Z_\text{large}\xspace}
\newcommand{\zzsmall}{Z_\text{small}\xspace}
\NewDocumentCommand\subsums{m+g}{
      \IfNoValueTF{#2}
          {\mathcal{S}(#1)}
          {\mathcal{S}(#1,#2)}\xspace
}
\newcommand{\oracle}{$(\eps,t)$-membership-oracle\xspace}
\newcommand{\oracleO}{$(\Oh(\eps),t)$-membership-oracle\xspace}

\newcommand{\sequence}[3]{$#1[#2,\ldots,#3]$\xspace}
\newcommand{\ssum}{\textsc{Subset Sum}\xspace}
\newcommand{\subsetsum}{\ssum}
\newcommand{\partition}{\textsc{Partition}\xspace}
\newcommand{\uknapsack}{\textsc{Unbounded Knapsack}\xspace}
\newcommand{\knapsack}{\textsc{Knapsack}\xspace}
\newcommand{\minconv}{$(\min,+)$-convolution\xspace}
\newcommand{\maxconv}{$(\max,+)$-convolution\xspace}
\newcommand{\sparsity}{\textsc{TreeSparsity}\xspace}
\newcommand{\IIIsum}{\textsc{$3$SUM}\xspace}
\newcommand{\ksum}{\textsc{$k$-SUM}\xspace}

\newcommand{\bringmann}{Bringmann's~\cite{bringmann-soda} algorithm\xspace}

\newcommand{\galil}{Galil~and~Margalit~\cite{galil-icalp} algorithm\xspace}
\newcommand{\kellerer}{Kellerer~et~al.~\cite{kellerer-subsetsum} algorithm\xspace}

\newcommand{\weak}{weak $(1-\eps)$-approximation\xspace}

\newcommand{\footremember}[2]{%
    \footnote{#2}
    \newcounter{#1}
    \setcounter{#1}{\value{footnote}}%
}
\newcommand{\footrecall}[1]{%
    \footnotemark[\value{#1}]%
} 

\newcommand{\defproblem}[3]{
  \vspace{2mm}
  \vspace{1mm}
\noindent\fbox{
  \begin{minipage}{0.95\textwidth}
  #1 \\
  {\bf{Input:}} #2  \\
  {\bf{Task:}} #3
  \end{minipage}
  }
  \vspace{2mm}
}

\title{A Subquadratic Approximation Scheme for Partition}
\date{}

\author{%
    Marcin Mucha\footremember{MIMUW}{Institute of Informatics, University of
    Warsaw, Poland, \texttt{\{mucha, k.wegrzycki, m.wlodarczyk\}@mimuw.edu.pl}}
    \and
    Karol W\k{e}grzycki\footrecall{MIMUW}
    \and
    Micha\l{} W\l{}odarczyk\footrecall{MIMUW}
}

\begin{document}

\maketitle

\thispagestyle{empty}
\begin{abstract}
    The subject of this paper is the time complexity of approximating \knapsack,
    \ssum, \partition, and some other related problems. The main result is an
    $\Ot(n+1/\eps^{5/3})$ time randomized FPTAS for \partition, which is derived
    from a certain relaxed form of a randomized FPTAS for \ssum. To the best of our knowledge,
    this is the first NP-hard problem that has been shown to admit a subquadratic time approximation scheme,
    i.e., one with time complexity of $\Oh((n+1/\eps)^{2-\delta})$ for some $\delta>0$.
    To put these developments in context, note that a quadratic FPTAS for
    \partition has been known for 40 years.
    
    Our main contribution lies in designing a mechanism that reduces an instance of \ssum to several simpler instances, each with some special structure, and keeps track of interactions between them.
    This allows us to combine techniques from approximation algorithms, pseudo-polynomial algorithms, and additive combinatorics.

    We also prove several related results.
    Notably, we 
    improve approximation schemes for \IIIsum, \minconv, and \sparsity.
    Finally, we argue why breaking the quadratic barrier for approximate \knapsack is unlikely by giving an $\Omega((n+1/\eps)^{2-o(1)})$
    conditional lower bound.
\end{abstract}

\clearpage
\setcounter{page}{1}

\section{Introduction}

The \knapsack-type problems are among the most fundamental optimization
challenges. These problems have been studied for more than a century already, as their origins can be traced 
back to the 1897's paper by Mathews~\cite{first-knapsack}. 




The \knapsack problem is defined as follows:
\begin{definition}[\knapsack]
    Given a set of $n$ items $E_n = \{1,\ldots,n\}$, with item $j$ having a positive
    integer weight $w_j$ and value $v_j$, together with knapsack capacity $t$.
    Select a subset of items $E \subseteq E_n$, such that the corresponding
    total weight $w(E) = \sum_{i \in E} w_i$ does not exceed the capacity $t$ and
    the total value $v(E) = \sum_{i \in E} v_i$ is maximized. 
\end{definition}

\knapsack is one of the 21 problems featured in Karp's list of NP-complete problems~\cite{karp21}.  
We also study the case where we are allowed to take each element multiple times, called \uknapsack.
Let $\Sigma(S)$ denote the sum of elements $S$. \ssum is defined as follows:

\begin{definition}[\ssum]
    \label{def:ssum}
    Given a set $S\subset \mathbb{N}$ of $n$ numbers (sometimes
    referred to as items) and an integer $t$,
    find a subset $S'\subseteq S$ with maximal $\Sigma(S')$ that does
    not exceed $t$.
\end{definition}
\ssum is a special case of \knapsack, where item weights 
are equal to item values. This problem is NP-hard as well. In fact, it remains NP-hard
even if we fix $t$ to be $\Sigma(S)/2$. This problem is called the \textit{Number Partitioning Problem} 
(or \partition, as we will refer to~it):

\begin{definition}[\partition]
    Given a set $S\subset \mathbb{N}$ of $n$ numbers,
    find a subset $S'\subseteq S$ with maximal $\Sigma(S')$ not exceeding
    $\Sigma(S)/2$.
\end{definition}


The practical applications of \partition problem range from
scheduling~\cite{partition1} to minimization of circuits sizes,
cryptography~\cite{partition2}, or even game theory~\cite{partition3,partition4}.
The decision version of this problem is sometimes humorously
referred to as ``the easiest NP-complete problem''~\cite{partition3}. In this paper we
will demonstrate that there is a grain of truth in this claim.

All the aforementioned problems are weakly NP-hard and admit pseudo-polynomial time algorithms.
The first such an algorithm for the \knapsack was proposed by~\citet{bellman}
and runs in time $\Oh(nt)$. This bound was improved for
the \ssum~\cite{koiliaris-soda} and the current best (randomized) time complexity for this problem
is $\Otilde(n+t)$, due to \citet{bringmann-soda} (for more on these and related results
see Section~\ref{sec:techniques}).  
The strong dependence on $t$ in all of these algorithms makes them
impractical for a large $t$ (note that $t$ can be exponentially larger than
the size of the input). This dependence has been shown necessary  as
an $\Oh\big(\poly(n)t^{0.99}\big)$ algorithm for the \ssum would contradict
both the SETH~\cite{subset-seth} and the SetCover conjecture~\cite{hard-as-cnf-sat}. 

One possible approach to avoid the dependence on $t$ is to settle for approximate solutions.
The notion of approximate solution we focus on in this paper is that of a \textit{Polynomial 
Time Approximation Scheme} (PTAS). A PTAS for a maximization problem is an algorithm that, 
given an instance of size $n$ and a parameter $\eps>0$, returns a solution with value $S$,
such that $\text{OPT}(1-\eps) \le S \le \text{OPT}$. It also needs to run in time polynomial
in $n$, but not necessarily in $1/\eps$ (so, e.g., we allow time complexities like $\Oh(n^{1/\eps})$). 
A PTAS is a \textit{Fully Polynomial Time Approximation Scheme} (FPTAS) if 
it runs in time polynomial in both $n$ and $1/\eps$. Equivalently, one can require the running time
to be polynomial in $(n+1/\eps)$. For example, $\Oh(n^2/\eps^4) = \Oh((n+1/\eps)^6)$. 
For definitions of problems in both exact and approximate sense see
Appendix~\ref{problems-definitions}.

The first approximation scheme for \knapsack (as well as \ssum and \partition as special cases) dates
back to 1975 and is due to~\citet{ibarra-kim}.  Its running time 
is $\Oh(n/\eps^2)$. After a long line of
improvements~\cite{levner79,levner94,fptas-knapsack,karp75,lawler79,knapsack-book},
the current best algorithms for each problem are:
the $\Oh(\min\{n/\eps, n+1/\eps^2\})$ algorithm for \partition due to~\cite{levner80},
the $\Oh(\min\{ n/\eps, n+1/\eps^2\log{(1/\eps)} \})$ algorithm for \ssum due to
\cite{kellerer-subsetsum} and, a very recent $\Ot(n+1/\eps^{{12}/{5}})$ for \knapsack, due to~\cite{chan-knapsack}.


Observe that all of these algorithms work in $\Omega((n+1/\eps)^2)$ time. In
fact, we are not aware of the existence of any FPTAS for an NP-hard problem
working in time $\Oh((n+1/\eps)^{2-\delta})$. 

\begin{problem} Can we get an $\Oh((n+1/\eps)^{2-\delta})$ FPTAS for any
    \knapsack-type problem (or any other NP-hard problem) for some constant
    $\delta>0$ or justify that it is unlikely?
\end{problem}

In this paper we resolve this question positively, by presenting the first such
algorithm for the \partition problem. This improves upon almost 40 years old
algorithm by~\citet{levner80}. On the other hand,
we also provide a conditional
lower bound suggesting that similar improvement for the more general \knapsack
problem is unlikely.

After this paper was announced, Bringmann~\cite{bringmann-com} showed that for 
any $\delta > 0$, an $\Oh((n+1/\eps)^{2-\delta})$ algorithm for \ssum would contradict the
\minconv-conjecture. This not only shows a somewhat surprising separation between
the approximate versions of \partition and \ssum, but also explains why our techniques
do not seem to transfer to approximating \ssum.

%

\subsection{Related Work}
\label{related-work}
In this paper we avoid the dependence on $t$ by settling on approximate instead of exact solutions.
Another approach is to allow running times exponential in $n$. This line of research has been 
very active with many interesting results.
The naive algorithm for \knapsack works in $\Os(2^n)$ time by simply
enumerating all possible subsets. \citet{meetinthemiddle} introduced the meet-in-the-middle approach and
gave an exact $\Os(2^{n/2})$ time and space algorithm. \citet{schroeppel}
improved the space complexity of that algorithm to $\Os(2^{n/4})$.
Very recently \citet{bansal17} showed an $\Os(2^{0.86n})$-algorithm working in polynomial space.

An interesting question (and very relevant for applications in cryptography) is how hard \knapsack type
problems are for random instances. For results in this line of research see~\cite{random-ss1,random-ss2,random-ss3,random-ss4}. 

\subsection{History of Approximation Schemes for \knapsack-type problems}
\label{subsec:history}

To the best of our knowledge, the fastest approximation for \partition dates back to
1980~\cite{levner80} with $\Ot(\min\{n/\eps, n+1/\eps^2\})$ running time\footnote{ 
As is common for \knapsack-type problems, the
$\Otilde$ notation hides terms poly-logarithmic in $n$ and $1/\eps$, but not in $t$.
}.
The majority of later research focused on
matching this running time for the \knapsack and \ssum.
In this section we will present an overview of the history of the FPTAS
for these problems.

\begin{table}[ht!]
    \centering

    \caption{Brief history of FPTAS for \knapsack-type problems. Since \partition is a special case
        of \ssum, and \ssum is a special case of \knapsack, an algorithm
        for \knapsack also works for \ssum and \partition.
        We omit redundant
        running time factors for clarity, e.g., \cite{kellerer-subsetsum}
        actually runs in $\Ot(\min\{n/\eps, n+1/\eps^2\})$ time but
        \cite{levner78,levner79} gave $\Oh(n/\eps)$ algorithm earlier.
        For a more robust history see~\cite[Section 4.6]{knapsack-book}. A star
        (*) marks the papers that match the previous best $\Ot\left( (n+1/\eps)^2 \right)$
    complexity for \partition problem.}
        
    \begin{tabular}{|l|l|r|}

        \hline 
        \textbf{Running Time} & \textbf{Problem} & \textbf{Reference} \\
        \hline
         $\Oh(n^2/\eps)$ & \knapsack & \cite{bellman} \cite{knapsack-book}\\
        \hline 
         $\Oh(n/\eps^2)$ & \knapsack & \cite{ibarra-kim,karp75} \\
        \hline 
         $\Oh(n/\eps)$ & \subsetsum  & * \cite{levner78,levner79} \\
        \hline
         $\Oh(n + 1/\eps^4)$ & \knapsack & \cite{lawler79} \\
        \hline
         $\Oh(n + 1/\eps^2)$ & \partition & * \cite{levner80} \\
        \hline
         $\Oh(n+1/\eps^3)$ & \ssum & \cite{levner94} \\
        \hline
         $\Ot(n+1/\eps^2)$ & \ssum & * \cite{kellerer-subsetsum} \\
        \hline
         $\Ot(n+1/\eps^{12/5})$ & \knapsack & * \cite{chan-knapsack} \\
        \hline
         $\Ot(n+1/\eps^{5/3})$ & \partition & \textbf{This Paper} \\
        \hline
    \end{tabular}
    \label{tab:history-subsetsum}
\end{table}

The first published FPTAS for the \knapsack is due to~\citet{ibarra-kim}. This naturally 
gives approximations for the \ssum and \partition as special cases. In their
approach, the items were partitioned into large and small classes. The profits are
scaled down and then the problem is solved optimally with dynamic programming. Finally,
the remaining empty space is filled up greedily with the small items. This algorithm
has a complexity $\Oh(n/\eps^2)$ and requires $\Oh(n+1/\eps^3)$
space.\footnote{In \cite[Section 4.6]{knapsack-book} there are claims, that
\citeyear{karp75} \citet{karp75} also gives $\Oh(n/\eps^2)$ approximation for
\ssum.} \citet{lawler79} proposed a different method of scaling and obtained $\Oh(n+1/\eps^4)$ running time.

Later, \citet{levner78,levner79} obtained an $\Oh(n/\eps)$
algorithm for the \ssum based on a different technique.
Then, in 1980 they proposed an even faster $\Oh(\min\{n/\eps,
n+1/\eps^2\})$ algorithm~\cite{levner80} for the \partition. To the best of our
knowledge this algorithm remained the best (until this paper).  
Subsequently, \citet{levner94} managed to generalize their result to \ssum with
an increase of running time and obtained
$\Oh(\min\{n/\eps,n+1/\eps^3\})$ time and $\Oh(\min\{n/\eps, n+1/\eps^2\})$
space algorithm~\cite{levner94}. Finally, \citet{kellerer-subsetsum} improved this algorithm for \ssum by giving 
$\Oh(\min\{ n/\eps, n+1/\eps^2\log{(1/\eps)} \})$ time and $\Oh(n+1/\eps)$ space
algorithm. This result matched (up to the polylogarithmic factors) the running
time for \partition.

For the \knapsack problem \citet{fptas-knapsack} gave an $\Oh(n $ $\min\{\log{n},
\log{(1/\eps)} \} $ $+ 1/\eps^2 \log{(1/\eps)} $ $\min \{n, 1/\eps \log{(1/\eps)}
\})$ time algorithm (note that the exponent in the parameter $(n+1/\eps)$ is 3 here) and
for \uknapsack \citet{fptas-uknapsack} gave an $\Oh(n+1/\eps^2 \log^3{(1/\eps)})$ time
algorithm (the exponent in $(n+1/\eps)$ is 2, see Appendix~\ref{problems-definitions} for the definition of \uknapsack). Very recently~\citet{chan-knapsack}
presented the currently best $\Ot(n + 1/\eps^{12/5})$ algorithm for the \knapsack.


\subsection{Our Contribution}

Our main result is the design of the mechanism that allows us to merge the
pseudo-polynomial time algorithms for \knapsack-type problems with algorithms on dense \ssum
instances. The most noteworthy application of these reductions is the following.

\begin{theorem}
    There is an $\Ot(n+1/\eps^{\frac{5}{3}})$ randomized time FPTAS for \partition.
\end{theorem}

This improves upon the previous, 40 year old bound of $\Ot(n + 1/\eps^2)$ for this problem, due to~\citet{levner80}. 
Our algorithm also generalizes to a \weak for \ssum.\footnote{Weak approximation can break the capacity constraint by a small factor. Definition~\ref{def:weak-ssum}
specifies formally what \weak for \ssum is.}

\begin{theorem}
    \label{weak-linear-ssum}
    There is a randomized \weak algorithm for \ssum running in $\Otilde\big(n + 1/\eps^{\frac{5}{3}}\big)$ time.
\end{theorem}

For a complete proof of these theorems see Section~\ref{fptas-ssum}.
We also present a conditional lower bound for \knapsack and \uknapsack. 

\begin{theorem}
    \label{knapsack-lower-bound}
    For any constant $\delta > 0$, an FPTAS for \knapsack or \uknapsack with
    $\Oh((n+1/\eps)^{2-\delta})$ running time would refute the \minconv conjecture.
\end{theorem}

This means that a similar improvement is unlikely for
\knapsack. Also, this shows that the algorithm of~\cite{fptas-uknapsack} for \uknapsack is 
optimal (up to polylogarithmic factors). This lower bound is relatively straightforward
and follows from previous
works~\cite{icalp2017,kunnemann-icalp2017} and was also observed in
\cite{chan-knapsack}. The lower bound  also
applies to the relaxed, \weak for \knapsack and \uknapsack, which separates these problems from \weak approximation for \ssum. This result was recently extended by Bringmann~\cite{bringmann-com} who showed a conditional hardness for obtaining a strong subquadratic approximation for \ssum, which explains why we need to settle for a weak approximation.

Lately it has been shown that the exact pseudo-polynomial algorithms for 
\knapsack and \uknapsack are subquadratically equivalent to the \minconv~\cite{icalp2017,kunnemann-icalp2017}.
Therefore, as a possible 
first step towards obtaining an improved FPTAS for \knapsack, we 
focus our attention on \minconv.

\begin{theorem}
    $(1+\eps)$-approximate \minconv can be computed in $\Ot((n/\eps) \log{W})$ time.
\end{theorem}

This also entails an improvement for the related \sparsity problem
(see Section~\ref{sec:tree-sparsity}).
The best previously known algorithms for both problems worked in time
$\Ot((n/\eps^2) \polylog(W))$ (see~\citet{tree-sparsity}).

The techniques used to improve the approximation algorithm for \minconv
also apply to approximation algorithms for \IIIsum.
For this problem we are able to show an algorithm that matches its asymptotic lower bounds.

\begin{theorem}
    There is a deterministic algorithm for $(1+\eps)$-approximate \IIIsum
    running in time $\Ot((n + 1/\eps)\polylog(W))$.
\end{theorem}

\begin{theorem}
    Assuming the Strong-3SUM conjecture, there is no $\Ot((n+1/\eps^{1-\delta})\polylog(W))$
    algorithm for $(1+\eps)$-approximate \IIIsum, for any constant $\delta > 0$.
\end{theorem}

For proofs of these theorems and detailed running times see
Sections~\ref{sec:ksum}

\subsection{Organization of the Paper}

In the Section~\ref{sec:techniques} we present the building blocks of our
framework and a sketch of the approximation scheme for \partition.
Section~\ref{prelim} contains the notation and preliminaries, and the main proof
is divided into Sections~\ref{sec:preprocessing} and~\ref{fptas-ssum}. In
Sections~\ref{approximate-minconv} and \ref{sec:tree-sparsity} we present the
algorithms for \minconv and \sparsity. In the Section~\ref{sec:ksum} we present
the algorithms for \IIIsum.   The proofs of technical lemmas can be found in
Appendix~\ref{proof-partition-reduction} and
Appendix~\ref{proof-numbers-rounding}. In Appendix~\ref{problems-definitions} we
give formal definitions of all problems. 
\section{Connecting Dense, Pseudo-polynomial and Approximation Algorithms for
\knapsack-type problems: An Overview}
\label{sec:techniques}

In this section we describe main building blocks of our framework. We also briefly discuss the recent advances in
the pseudo-polynomial algorithms for \ssum and discuss how to use them. Then, we explain the 
intuition behind the trade-off we exploit and give a sketch of the main algorithm.
The formal arguments are located in Section~\ref{fptas-ssum}.





\paragraph{Difficulties with Rounding for \ssum}

There is a strong connection between approximation schemes and pseudo-polynomial algorithms~\cite{woeginger}.
For example, a common theme in approximating knapsack is to reduce the range of the \emph{values} 
(while keeping the \emph{weights} intact) and then apply a pseudo-polynomial algorithm.
Rounding the weights would be tricky because of the hard knapsack constraint. 
In particular, if one rounds the weights down, some feasible solutions
to the rounded instance might correspond to infeasible solutions in the original instance. 
On the other hand, when rounding up, some feasible solutions might become infeasible in the 
rounded instance.

Recently, new pseudo-polynomial algorithms have been proposed for  \ssum
(see~\citet{koiliaris-soda} and \citet{bringmann-soda}). A natural idea is to use these 
to design an improved approximation scheme for \ssum. However, this seems to be difficult
due to rounding issues discussed above. After this paper was announced,
Bringmann~\cite{bringmann-com} explained this difficulty by giving
a conditional lower bound on a quadratic approximation of \ssum. 



\subsection{Weak Approximation for \ssum and Application to \partition}
Because of these rounding issues, it seems hard to design a general rounding scheme that, 
given a pseudo-polynomial algorithm for \ssum, produces an FPTAS for \ssum. What
we can do, however, is to settle for a weaker notion of approximation.
 
\begin{definition}[Weak apx for \ssum]
    \label{def:weak-ssum}
    Let $Z^*$ be the optimal value for an instance $(Z,t)$ of \ssum. Given $(Z,t)$, a
    \textit{weak} $(1-\eps)$-approximation algorithm for \ssum returns $Z^H$ such that
    $(1-\eps)Z^* \le Z^H < (1+\eps)t$.
\end{definition}
Compared to the traditional notion of approximation, here we allow a small violation of the
packing constraint. This notion of approximation is interesting in itself. Indeed,
it has been already considered in the stochastic regime for \knapsack~\cite{stochastic-knapsack}. 

Before going into details of constructing the \weak algorithms for the \ssum, let us
establish a relationship with the approximation for the \partition.

\begin{observation}
    \label{partition-reduction}
    If we can weakly $(1-\eps)$-approximate \ssum in time $\Ot(T(n,\eps))$, then we
    can $(1-\eps)$-approximate \partition in the same $\Ot(T(n,\eps))$ time.
\end{observation}

This is because of the symmetric structure of \partition problem: If a subset
$Z'$ violates the hard constraint ($t \le \Sigma(Z') \le (1+\eps)t$), then the
set $Z-Z'$ is a good approximation and does not violate it (recall that in
\partition problem we always have $t = \Sigma(Z)/2$). For a formal proof see
Section~\ref{proof-partition-reduction}.

\subsection{Constructing Weak Approximation Algorithms for \ssum: A Sketch}
\label{sec:weak}
\begin{fact}
    \label{thm:weak-example}
    Given an $\Ot(T(n,t))$ exact algorithm for \ssum, we can construct a \weak
    algorithm for \ssum working in time $\Ot(T(n,\frac{n}{2\eps}))$.
\end{fact}

\begin{proof}

    We assume that the exact algorithm for the \ssum
    works also for multisets. We will address this issue in more detail in Section~\ref{sec:multisets}.

    Let $Z=\{v_1,\ldots,v_n\}$ and $t$ constitute a \ssum instance. Let $I$ be the set of indices
    of elements of some optimal solution, and let OPT be their sum. Let us also introduce a scaled
    approximation parameter $\eps' = \frac{\eps}{4}$.  

    Let $k=\frac{2 \eps' t}{n}$. Define a rounded instance as follows: the (multi)-set of $\tilde{Z}$ contains 
    a copy of $\tilde{v}_i = \floor{\frac{v_i}{k}}$ for each $i \in
    \{1,\ldots,n\}$, and $\tilde{t}=\floor{\frac{t}{k}}$.
    
    Apply the exact algorithm $\mathcal{A}$ to the rounded instance $(\tilde{Z},\tilde{t})$. Let $I'$ be the set of indices of elements of
    the solution found.

    We claim that $\{v_i : i \in I'\}$  is a weak $(1-\eps)$ approximation for
    $Z$ and $t$. First let us show that 
    this solution is not much worse than OPT:

    \begin{displaymath}
         \sum_{i \in I'} v_i \ge k\sum_{i \in I'} \tilde{v}_i \ge k\sum_{i \in I} \tilde{v}_i = k\sum_{i \in I} \floor{\frac{v_i}{k}} 
        \ge \sum_{i \in I} (v_i - k) 
            \ge \text{OPT} - nk = \text{OPT} - 2\eps' t \ge \text{OPT}(1-\eps).
    \end{displaymath}

     The last inequality holds because we can assume $\text{OPT} \ge t/2$ (see
     Section~\ref{items-partition} for details).

     Similarly, we can show that this solution does not violate the hard constraint by too much:
     \begin{displaymath}
         \sum_{i \in I'} v_i \le \sum_{i \in I'} (k\tilde{v}_i+k) \le nk+k\sum_{i \in I'} \tilde{v}_i \le nk+\tilde{t}k 
         \le nk+k+t \le 3\eps' t + t \le t(1+\eps).
     \end{displaymath}

    Finally, since the exact algorithm is applied to a (multi)-set of $n$ items with $\tilde{t}=\floor{\frac{t}{k}} = \floor{\frac{n}{2\eps'}}$,
    the resulting algorithm runs in the claimed time.
\end{proof}

We state the above proof only to give the flavour of the basic form of 
reductions in this paper. Usually reductions that we will consider are more complex for technical reasons. 
One thing to note in particular is that the relation between $k$ and $\eps$ is dictated by the fact, that there may be as many as $n$
items in the optimal solution. Given some control over the solution size, one
can improve this reasoning (see Lemma~\ref{exact-approx}).


\subsection{Approximation via Pseudo-polynomial time \ssum algorithm}

Currently, the fastest pseudo-polynomial algorithm for \ssum runs 
in time $\Ot(n+t)$, randomized. $\subsums{Z}{t}$ denotes the set of all
possible subsums of set $Z$ up to integer $t$ (see Section~\ref{prelim}).

\begin{theorem}[\citet{bringmann-soda}]
    \label{thm:bringmann}
    There is a randomized, one-sided error algorithm with running time
    $\Oh(n + t\log{t}\log^3{\frac{n}{\delta}}\log{n})$, that returns a set $Z'
    \subseteq \subsums{Z}{t}$, containing each element from
    $\subsums{Z}{t}$ with probability at least $1-\delta$.
\end{theorem}

This suffices to solve \ssum exactly with high probability.  Here $\subsums{Z}{t}$ is represented by a binary
array which for a given index $i$ tells whether there is a subset that sums up
to $i$ (see Section~\ref{prelim} for a formal definition).
For our trade-off, we actually need a probabilistic guarantee on all elements of $\subsums{Z}{t}$ simultaneously.
Fortunately, this kind of bound holds for this algorithm as well
(see~\cite[Appendix B.3.2]{icalp2017} for detailed analysis).

\begin{corollary}\label{cor:bringmann-all}
    There is a randomized $\Ot(n+t)$ algorithm that computes $\subsums{Z}{t}$ with a constant probability of success.
\end{corollary}


The first case where this routine comes in useful occurs when all items are in
the range $[\gamma t,t]$ (think of $\gamma$ as a trade-off
parameter set to $\eps^{-2/3}$). Note, that any solution summing to at most $t$ can
consist of at most $1/\gamma$ such elements. This observation allows us to
round the elements with lower precision and still maintain a good approximation ratio, as follows:

\begin{align*}
    v'_i  =  \floor{\frac{2v_i}{\gamma \eps t}},
    \;\;\;\;\;\;\;\;
    t' =  \floor{\frac{2t}{\gamma \eps t}}  = \floor{\frac{2}{\gamma \eps}}
    .
\end{align*}

\bringmann on the rounded instance runs in time $\Ot(n + t') = \Ot(n +
\frac{1}{\gamma \eps})$ and returns an array of solutions with an additive error
$\pm \eps t$ with high probability (see Lemma~\ref{lem:large-algorithm}). 
Similar reasoning about sparseness also applies if the number of items is 
bounded (i.e., when $n = \Ot(\frac{\gamma}{\eps})$). In that case \bringmann
runs in time $\Ot(\frac{\gamma}{\eps^2})$ and provides the same guarantees (see Lemma~\ref{small-sparse} and also the next section).

\subsection{Approximation via Dense \ssum}
\label{galil-box}

Now we need a tool to efficiently solve the instances where all items are in
range $[0,\gamma t)$, so-called dense instances. More
formally, an instance consisting of $m$ items is dense if all items are in the range
$[1,m^{\Oh(1)}]$. Intuitively, rounding does not work well for these instances
since it introduces large rounding errors.
On the other hand, if an instance contains many distinct numbers on a small interval, one can exploit its additive structure.

\begin{theorem}[\citet{galil-icalp}]
    \label{galil-algorithm}
    Let $Z$ be a set of $m$ distinct numbers in the interval $(0,\ell]$ such that

    \begin{displaymath}
        m > 1000\cdot\sqrt{\ell}\log{\ell},
    \end{displaymath}
    and let $L := \frac{100\cdot\Sigma(Z) \ell \log{\ell}}{m^2}$.

    Then in $\Oh(m + ((\ell/m)\log{\ell})^2)$ preprocessing time we can build a
    structure that can answer the following queries in constant time. In a query the structure receives a
    target number $t \in (L, \Sigma(Z)-L)$ and decides whether there is a $Z' \subseteq Z$
    such that $\Sigma(Z') = t$. The structure is deterministic.
\end{theorem}

In fact we will use a more involved theorem that can also construct a
solution in $\Oh(\log(l))$ time but we omit it here to keep this section
relatively free of technicalities (see Section~\ref{smallitems} for a
discussion regarding these issues).

Observe that $L = \Ot(\ell^{1.5})$ (because $\Sigma(Z) < m\ell$) and
the running time is bounded by $\Ot(m+\ell)$ (because
$\ell / m = \Oh(\sqrt{\ell})$).
We will apply this result for the case $\ell = \gamma t$ (see Lemma~    \ref{lem:small-algorithm}).
Recall,
that \bringmann runs in time $\Ot(m+t)$, which would be slower by the factor $\gamma$ (the trade-off parameter).
For simplicity, within this overview we will assume, that Theorem~\ref{galil-algorithm}
 provides a data structure that can answer queries with the
target numbers in $[0, \Sigma(Z)]$. In the actual proof, we need to overcome this obstacle, by
merging this data structure with other structures, responsible for targets near the boundary, which we call
\emph{marginal} targets
(see Lemma~\ref{small-range-bringmann}).

Suppose our instance consists of $m$ elements in the range $[0,\gamma t]$. We use the 
straightforward rounding scheme, as in the proof of Fact~\ref{thm:weak-example}.

\begin{align*}
    v'_i =  \floor{\frac{2m v_i}{\eps t}},
    \;\;\;\;\;\;\;\;
    t' =  \floor{\frac{2 m t}{\eps t}} =  \floor{\frac{2m}{\eps}}.
\end{align*}

We chose $\gamma t$ as the upper bound on item size, so that $\ell' = m\gamma / \eps$ is an
upper bound on $v'_i$. Now, if the number of items satisfies the inequality
$\ell' < m^2$, then we can use the~Theorem~\ref{galil-algorithm} with running
time $\Ot(m+\ell') = \Ot(m + m\gamma / \eps)$. This provides a data
structure that can answer queries from the range that is of our interest (for
a careful proof see Section~\ref{fptas-ssum}).

Still, it can happen that most of the items are in the sparse instance
(i.e., $\ell'\ge m^2$) and we cannot use the approach from~\cite{galil-icalp}. In that case we
use Theorem~\ref{thm:bringmann}
 again, with running time $\Ot(m + \frac{\gamma}{\eps^2})$ (see Lemma~\ref{small-sparse}).

In the end, we are able to compute an array of solutions, for items in
range $[0, \gamma t]$ in time $\Ot(m + \frac{\gamma}{\eps^2} +
\frac{m\gamma^2}{\eps^2})$ with additive error $\pm \eps t$ and high probability (see Lemma~\ref{small-algorithm}).
The last term in time complexity comes from handling the marginal queries.

\subsection{A Framework for Efficient Approximation}
\label{approach}

In this section we will sketch the components of our mechanism (see
Algorithm~\ref{alg:roadmap}). The mechanism combines pseudo-polynomial
\bringmann with \galil for dense instances of \ssum. 

\begin{algorithm}
    \caption{Roadmap for the \weak for \ssum. Input: item set $Z, t, \eps$}
	\label{alg:roadmap}
\begin{algorithmic}[1]
	\State ensure $OPT \ge t/2$ 
	\State reduce $|Z|$ to $\Otilde({1}{/\eps})$ 
	\Repeat
		\State partition items into $\zzlarge$ and $\zzsmall$
		\State divide $[0, \gamma t]$ into $\ell = \Oh(\gamma\log (n) / \eps)\cdot |\zzsmall|$ segments
        \State round down small items
		\State remove item repetitions in $\zzsmall$ 
	\Until {$\ell = \Oh(\gamma\log (n) / \eps)\cdot |\zzsmall|$} 
	    \State build a data structure for large items 
	\If {$|\zzsmall| = \Otilde(\sqrt{\ell})$}
         \State build a data structure for small items 
      \Else
      	\State build data structures for marginals 
      	\State exploit the density of the instance to cover the remaining case 
      \EndIf
    \State merge the data structures for large and small items 
\end{algorithmic}
\end{algorithm}

We begin by reducing the number of items in the instance $Z$ to roughly
$\Ot(1/\eps)$ items to get a near linear running time (see
Lemma~\ref{lem:reduce-items}). After that our goal is to divide items into
\emph{small} and \emph{large} and process each part separately, as described earlier. 

However, Theorem~\ref{galil-algorithm} requires a lower bound on the number of \textbf{distinct} items. 
To control this parameter, we  merge identical items into
larger ones, until each item appears at most twice. 
However, this changes the number of items, and so the procedure might have to be restarted.
Lemma~\ref{lem:partition-small-large} guarantees that
we require at most
$\log n$ such refinement steps.

In the next phase we decide which method to use to solve the instance depending
on its density (line 10). We encapsulate these methods into data
structures (lines 11-14).
Finally we will need to merge the solutions. For this
task we introduce the concept of membership oracles (see
Definition~\ref{mem-oracle}) that are based on FFT and backtracking to retrieve
solutions (see Lemma~\ref{FFT}).
The simplified trade-off schema is presented on the Figure~\ref{fig:overall}.


\begin{figure}[ht!]
    \centering
	\begin{tikzpicture}[xscale=1, node distance=1em and 1em]
       \def\widthrectangle{12}
       \def\heightrectangle{2}
       \def\tradeoffK{5}

       \node (end_rectangle) at (\widthrectangle,\heightrectangle) {};
       \node (k_up)   at (\tradeoffK,\heightrectangle) {};
       \node (k_down) at (\tradeoffK,0) {};

       \draw (0.5*\tradeoffK,0.75*\heightrectangle)    node () {Dense Instance};
       \draw (0.5*\tradeoffK + 0.5*\widthrectangle,0.75*\heightrectangle) node () {Sparse Instance};
       \node[above of=end_rectangle, distance=1em] {$t$};
       \node[above of=k_up] {$\gamma t$};

       \draw (0,0) rectangle (end_rectangle);
       \draw (k_up) -- (k_down);

       \foreach \x in {0,...,10}
         \draw (\tradeoffK + \x*7/10,0) -- (\tradeoffK + \x*7/10, \heightrectangle*0.3);
       
       \foreach \x in {0,...,15}
         \draw (\x*5/15,0) -- (\x*5/15, \heightrectangle*0.3);
       \draw[decoration={brace, raise=5pt},decorate] (4+1/3,0.6) -- node[above=5pt] {$\frac{\epsilon t}{m}$} (4+2/3,0.6);
       \draw[decoration={brace, raise=5pt},decorate] (\tradeoffK+7/10,0.6) -- node[above=5pt] {${\epsilon\gamma t}$} (\tradeoffK+14/10,0.6);

		\node (galil) at (1, -1.25) {$\Ot(\frac{\gamma}{\eps^2} +
\frac{m\gamma^2}{\eps^2})$ };
		\node (kellerer) at (5, -1.25) {$\Ot(m + \frac{\gamma}{\eps^2})$};
       \node (condition1) at (1,-0.25) {If $m^2 > \ell$};
       \node (condition2) at (4,-0.25) {$m^2 \le \ell$};

       \draw (\tradeoffK*0.5, -0.25) -- (\tradeoffK*0.5,-0.5);
       \draw[->, to path={-| (\tikztotarget)}] (\tradeoffK*0.5-1,-0.5) to (galil);
       \draw[->, to path={-| (\tikztotarget)}] (\tradeoffK*0.5-1,-0.5) to (kellerer);
            
       \def\xSparse{7}
		\node (bringmann) at (9.5, -1.25) {$\Ot(n+ \frac{1}{\gamma\epsilon})$};

       \draw (\tradeoffK*0.5+\xSparse, -0.25) -- (\tradeoffK*0.5 + \xSparse,-0.5);
       \draw[->, to path={-| (\tikztotarget)}] (\tradeoffK*0.5+\xSparse,-0.5) to (bringmann);

   \end{tikzpicture}
	\caption{Overall schema of trade-off and usage of building blocks. The parameter $m$ denotes number of items in the dense instance, $n$ is the number of all elements, $\gamma$ is the trade-off parameter, $\ell$ is the upper bound on the item size after rounding, $t$ is the target sum. The buckets in the sparse/dense instance depict the rounding scheme for small and large items.}
   \label{fig:overall}
\end{figure}
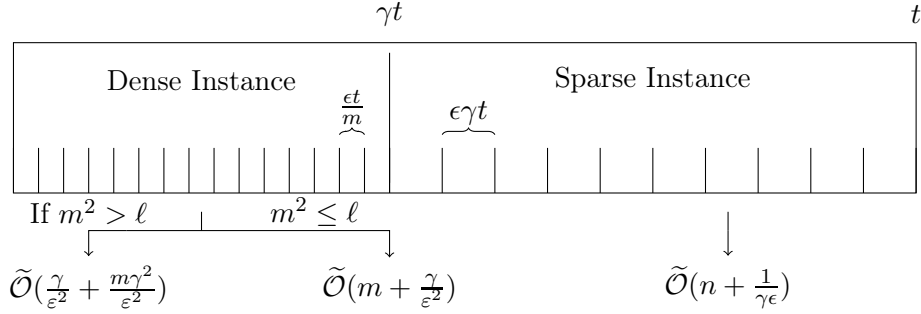


The final running time of our framework
is $\Ot(n + \frac{1}{\gamma\eps} + \frac{\gamma}{\eps^2} +
\frac{\gamma^2}{\eps^3} )$ with high probability for any $\gamma(n,\eps) > 0$
(see Lemma~\ref{weak-tradeoff}). For $\gamma = \eps^{-2/3}$, this  gives us
an $\Ot(n + \eps^{-5/3})$ time \weak approximation for \ssum.

\section{Preliminaries}
\label{prelim}

For a finite multiset $Z \subset \nat$ we denote its size as $|Z|$, the number
of distinct elements as $||Z||$, and the sum of its elements as $\Sigma(Z)$.

For a number $x$ we define $\pow(x)$ as the largest power of 2 not exceeding $x$.
If $x < 2$ we set $\pow(x) = 1$.

For sets $A,B \subset \nat$ their bounded algebraic sum $A \oplus_t B$ is a set
$\{a+b : a \in \{0\} \cup A,\, b \in \{0\} \cup B\} \cap [0,t]$.

\begin{definition}[Subsums]

For a finite multiset $Z \subset \nat$ we define $\subsums{Z}_k$ as a set of all
possible subset sums of $Z$ of size at most $k$, i.e., $x \in \subsums{Z}_k$ iff 
there exists $S' \subseteq Z$, such that $\Sigma(S') = x$ and $|S'| \le k$.
$\subsums{Z}$ is the set without the constraint on the size of the subsets,
i.e., $\subsums{Z}:=\subsums{Z}_\infty$.
The capped version is defined as $\subsums{Z}{t}_k := \subsums{Z}_k \cap [0,t]$
and $\subsums{Z}{t} := \subsums{Z} \cap [0,t]$.

We call two multisets $Z_1, Z_2 \subset \nat$ equivalent if
$\subsums{Z_1} = \subsums{Z_2}$.
\end{definition}

Note that $0 \in \subsums{Z}{t}_k$ for all sets $Z$ and $t,k>0$.

\begin{definition}[$(\eps,t)$-closeness]
We say that set $B$ is $(\eps,t)$-close to $A$ if there is a surjection $\phi:A
\rightarrow B$ such that $x -\eps t \le \phi(x) \le x + \eps t$.
A \ssum instance $(Z_2, t)$ is $\eps$-close to $(Z_1, t)$ if
$\subsums{Z_2}{t}$
is $(\eps,t)$-close to $\subsums{Z_1,t}$.
\end{definition}

Sometimes, when there is no other notation on $t$, we will use the notion of
$\eps$-closeness as a $(\eps,t)$-close.

Usually the surjection from the definitions will come by rounding down the item sizes
and each item set will get a moderately smaller total size.
We will also apply the notion of $(\eps,t)$-closeness to binary arrays having in mind the sets they represent.

\begin{fact}
If $A$ is $(\eps,t)$-close to $\subsums{Z_1}{t}$ and $B$ is $(\eps,t)$-close
to $\subsums{Z_2}{t}$
then $A \oplus_t B$ is $(2\eps,t)$-close to $\subsums{Z_1 \cup Z_2}{t}$
\end{fact}

We will also need to say, that there are no close elements in a set. It will
come in useful to show, that after rounding down all the elements are
distinct.
\begin{definition}[$(x)$-distinctness]
    The set $S$ is said to be $(x)$-distinct if every interval of length $x$ contains at most one item from $S$. 
     The set $S$ is said to be $(x,2)$-distinct if every interval of length $x$ contains at most two items from $S$. 
\end{definition}

\section{Preprocessing}
\label{sec:preprocessing}

This section is devoted to simplify the instance of \ssum in order to produce a
more readable proof of the main algorithm. In here we will deal with:
\begin{itemize}
    \item multiplicities of the items,
    \item division of the instance into large and small items,
    \item proving that rounding preserves $\eps$-closeness,
    \item reducing a number of items from $n$ to $\Ot(1/\eps)$ items.
\end{itemize}

The solutions to these problems are rather technical and well known in the
community~\cite{kellerer-subsetsum, knapsack-book, koiliaris-soda,
bringmann-soda}. We include it in here because these properties
are used in approximation algorithms~\cite{knapsack-book, kellerer-subsetsum}
and exact pseudo-polynomial algorithms~\cite{koiliaris-soda,bringmann-soda}
communities separately. We expect that reader may not be familiar with both of
these technical toolboxes simultaneously and accompany
this section with short historical references and pointers to the
original versions of proofs.

\subsection{From Multisets to Sets}
\label{sec:multisets}

The general instance of \ssum may consists of plenty of items with equal size. Intuitively, these instances seem to be much simpler than instances
where almost all items are different. The next lemma 
will allow us to formally capture this intuition with the appropriate
reduction. This lemma was proposed in~\cite[Lemma 2.2]{koiliaris-soda} but was also used
in~\cite{bringmann-soda}.

\begin{lemma}[cf. Lemma 2.2 from~\cite{koiliaris-soda}]\label{lem:multi-items}
Given a multiset $S$ of integers from $\{1,\ldots,t\}$, such that $|S|=n$ and the number of distinct items $||S||$ is $n'$, one can
compute, in $\Oh(n\log{n})$ time, a multiset $T$, such that:

\begin{itemize}
    \item $\subsums{S}{t} = \subsums{T}{t}$
    \item $|T| \le |S|$
    \item $|T| = \Oh(n' \log{n})$
    \item no element in $T$ has multiplicity exceeding two.
\end{itemize}
\end{lemma}
\begin{proof}
We follow the proof from \cite[Lemma 2.2]{koiliaris-soda},
however the claimed bound on $|T|$ is only $\Oh(n' \log{t})$ therein.
Consider an element $x$ with the multiplicity $2k+1$.
We can replace it with a single copy of $x$ and $k$ copies of $2x$
while keeping the multiset equivalent.
If the multiplicity is $2k+2$ we need 2 copies of $x$ and $k$ copies of $2x$.
We iterate over items from the smallest one and for each with at least 3 copies
we perform the replacement as described above.
Observe that this procedure generates only elements of form $2^ix$ where
$i \le \log{n}$ and $x$ is an element from $S$. This yields the bound on $|T|$.
The routine can be implemented to take $\Oh(\log n)$ time for creating
each new item using tree data structures.
\end{proof}

\subsection{From $n$ Items to $\Ot(1/\eps)$ Items}
\label{sec:reduction-items}

To reduce number of items $n$ to $\Ot(1/\eps)$ 
\citet{kellerer-subsetsum} gave a very intuitive construction that later found
 applications in~\knapsack-type problems~\cite{knapsack-book}.

Intuitively, rounding scheme described in
Section~\ref{sec:techniques} could divide the items into $\Oh(n/\eps)$ intervals and
this would result with an $\eps$-close instance to the original one. In here we start
similarly but we want to get rid of factor $\Oh(n)$.
We divide an instance to
$k = \ceil{\frac{1}{\eps}}$ intervals of length $\eps t$, i.e., $I_j := (jt,
(j+1)t]$. Next notice that for interval $I_j$ we do not need to store more than
$\Oh(\ceil{\frac{k}{j}})$ items, because their sum would exceed $t$ (this is
the step where $\eps$ factor will come in). Finally, the number of items is upper
bounded (up to the constant factors):

\begin{displaymath}
    \sum_{j=1}^{k} \ceil{\frac{k}{j}} \le k \sum_{j=1}^{k}  \frac{1}{j} < k
    \log{k} = \Oh(1/\eps \log{(1/\eps)})
\end{displaymath}

The last inequality is just an upper bound on harmonic numbers. This was a very
informal sketch of the proof of~\cite{kellerer-subsetsum} construction to give
some intuition.
The next technical lemma is based on their trick.

\begin{lemma}\label{lem:reduce-items}
Given a \ssum instance $(Z, t),\, |Z| = n$, one can find an $\eps$-close instance $(Z_2, t)$ such that $|Z_2| = \Oh\big(\frac{1}{\eps}\log(\frac{n}{\eps})\log(n)\big)$.
The running time of this procedure is $\Oh(|Z| + |Z_2|)$.
\end{lemma}
\begin{proof}
We begin with constructing $Z_1$ as follows.
For $i=1,\dots,\log(\frac{2n}{\eps})$ we round down each element in $Z \cap [\frac{t}{2^i}, \frac{t}{2^{i-1}})$ to the closest multiplicity of $\floor{\frac{\eps t}{2^{i+1}}}$.
We neglect elements smaller than $\frac{\eps t}{2n}$.
Observe that $||Z_1|| = \Oh\big(\frac{1}{\eps}\log(\frac{n}{\eps})\big)$.

We argue that $(Z_1, t)$ is $\eps$-close  to $(Z, t)$.
To see this, consider any subset $I \subseteq Z$ summing to at most $t$
and its counterpart $Y_1 \subseteq Z_1$.
We lose at most $n\cdot\frac{\eps t}{2n} = \frac{\eps t}{2}$ by omitting items smaller than $\frac{\eps t}{2n}$.
Let $k_i = |I \cap [\frac{t}{2^i}, \frac{t}{2^{i-1}})|$ and $t_i$ denote the sum of elements in $I \cap [\frac{t}{2^i}, \frac{t}{2^{i-1}})$.
Since each element in $[\frac{t}{2^i}, \frac{t}{2^{i-1}})$ has been decreased by at most $\frac{\eps t}{2^{i+1}}$ and $k_i\cdot \frac{t}{2^i} \le t_i$, we have
\begin{equation*}
\Sigma(I) - \Sigma(Y_1) \le \frac{\eps t}{2} + \sum_{i=1}^{\log(\frac{2n}{\eps})} k_i\cdot\frac{\eps t}{2^{i+1}} \le \frac{\eps t}{2} + \sum_{i=1}^{\log(\frac{2n}{\eps})} \frac{\eps t_i}{2} \le \eps t.
\end{equation*}

In the end we take advantage of Lemma~\ref{lem:multi-items} to transform $Z_1$
into an equivalent multiset $Z_2$ such that $|Z_2| \le ||Z_1|| \log(|Z_1|) =  \Oh\big(\frac{1}{\eps}\log(\frac{n}{\eps})\log(n)\big)$.
\end{proof}

Note, that we discarded items smaller than $\frac{\eps t}{2n}$. We do this
because sum of these elements is just too small to influence the worst case
approximation factor. We do not consider them just for the simplicity of
analysis. To make this algorithm competitive
in practice, one should probably just greedily add these small items to
get a little better solution. 
    
\subsection{From One Instance to Small and Large Instances}
\label{items-partition}

First we need a standard technical assumption, that says that we can cheaply
transform an instance to one with a lower bounded solution. We will need it
just to simplify the proofs (e.g., it will allow us to use Lemma~\ref{rounding-numbers-d} multiple
times).

\begin{lemma}\label{lem:large-opt}
One may assume w.l.o.g. that for any \ssum instance $\text{OPT} \ge \frac{t}{2}$.
\end{lemma}
\begin{proof}
Let us remove from the item set $Z$ all elements exceeding $t$ since they cannot belong to any solution.
If $\Sigma(Z) \le t$ then the optimal solution consists of all items.
Otherwise consider a process in which $Y_1 = Z$ and in each turn we obtain $Y_{k+1}$ by dividing $Y_k$ into two arbitrary non-empty parts and taking the one with a larger sum.
We terminate the process when $Y_{last}$ contains only one item.
Since $\Sigma(Y_1) > t,\, \Sigma(Y_{last}) \le t$, and in each step the sum
decreases by at most factor two, for some $k$ it must be $\Sigma(Y_k) \in [\frac{t}{2}, t]$.
Because there is a feasible solution of value at least $\frac{t}{2}$, $\text{OPT}$ cannot be lower.
\end{proof}


One of the standard  ways of solving \ssum is to separate the
large and small items~\cite{knapsack-book}. Usually these approximations consider items greater
and smaller than some trade-off parameter.  Our techniques
require a bound on the multiplicities of small items,
which is provided by the next lemma.

\begin{lemma}[Partition into Small / Large Items]
    \label{lem:partition-small-large}

    Given an instance $(Z,t)$ of \ssum, an approximation factor $\eps$, and
    a trade-off parameter $\gamma$, one can
    deterministically transform the instance $(Z,t)$, in time $\Oh(n\log^2{n})$, to an $\eps$-close instance $(\zzsmall \cup \zzlarge, t)$ such
    that:

    \begin{itemize}
        \item $\forall z_s \in \zzsmall, \;\; \text{it holds that} \;\; z_s < \gamma t$,
        \item $\forall z_l \in \zzlarge, \;\; \text{it holds that} \;\; z_l \ge \gamma t$,
        \item The set $\zzsmall$ is $(\frac{\eps t}{m \cdot \log n}, 2)$-distinct where $m = \Oh(|\zzsmall|)$, i.e., after rounding there can be at most 2 occurrences of each item.
    \end{itemize}
\end{lemma}

\begin{proof}


We call an item $x$ \emph{large} if $x \ge \gamma t$ and \emph{small} otherwise.
Let $Y_0$ be the initial set of small items and $m_0 = |Y_0|,\, q_0 = \pow(\frac{\eps t}{m_0\log{n}})$.
We round down the size of each small item to the closest multiplicity of $q_0$.
Then we apply Lemma~\ref{lem:multi-items} to the set of small items to get rid of items with 3 or more copies.
Note that this operation might introduce new items that are large.
We obtain a new set of small items $Y_1$ and repeat this procedure with notation $m_i = |Y_i|,\, q_i = \pow(\frac{\eps t}{m_i\log{n}})$.
It holds that $m_{i+1} \le m_i$ and $q_i\,|\,q_{i+1}$.
We stop the process when $m_{i+1} \ge \frac{m_i}{2}$, which means there can be at most $\log{n}$ iterations.
Let $m$ denote the final number of small items and $q \ge \frac{\eps t}{4m\log{n}}$ -- the last power of 2 used for rounding.
All small items now occur with multiplicities at most~2.

Let us fix $\zzsmall$ as the set of small items after the modification above
and likewise $\zzlarge$.
In the $i$-th rounding step values of $m_i$ items are being decreased by at most $\frac{\eps t}{m_i\log{n}}$,
so each new instance is $\frac{\eps}{\log n}$-close to the previous one.
There are at most $\log{n}$ steps and the removal of copies keeps the instance equivalent,
therefore $(\zzsmall \cup \zzlarge,\, t)$ is $\eps$-close to $(Z, t)$.
\end{proof}

Our algorithm works independently on these two instances and produces two
arrays $\eps$-close to them. The construction below allows us to join these
solutions efficiently. We want to use them even if we have only access
to them by queries. We  formalize this as an \oracle.
The asymmetry of the definition below will become clear in Lemma~    \ref{exact-approx}.

\begin{definition}[\oracle]
    \label{mem-oracle}
    The \oracle of a set $X$ is a data structure, that given an integer $q$ answers \textbf{yes/no} obeying following conditions:
    \begin{enumerate}
    \item if $X$ contains an element in $[q-\eps t, q+ \eps t]$, then the answer is \textbf{yes},
  \item if the answer was  \textbf{yes}, then $X$ contains an element in $[q-2\eps t, q+ 2\eps t]$.
       \end{enumerate}
A query to
   the oracle takes $\Ot(1)$ time.
    Moreover, if the oracle answers \textbf{yes}, then it can return a witness $x$ in $\Ot(1)$ time.
\end{definition}

Below we present an algorithm that can efficiently join the solutions. We
assume, that we have only query-access to them and want to produce an \oracle
of the merged solution.

\begin{lemma}[Merging solutions]
    \label{FFT}
    Given $\subsums{Z_1}{t}$ and $\subsums{Z_2}{t}$ as $(\eps,t)$-membership-oracles
    \begin{itemize}
        \item $S_1$ that is  $(\eps,t)$-close instance to $\subsums{Z_1}{t}$,
        \item $S_2$ that is $(\eps,t)$-close instance to $\subsums{Z_2}{t}$,
    \end{itemize}
    we can, deterministically in $\Ot(\frac{1}{\eps})$ time, construct a
    $(2\eps,t)$-membership-oracle for
    $\subsums{Z_1\cup Z_2}{t}$.
\end{lemma}

\begin{proof}
    For an ease of presentation, only in this proof we will use interval
    notation of inclusion, i.e., we will say that $(a,b] \sqcap A$ iff $\exists_x x\in
    (a,b] \wedge x \in A$.
    Let $p = \Oh(\eps t)$. 
    For each interval $(ip, (i+1)p]$ where $i \in \left\{0,\ldots
    \floor{\frac{t}{p}} \right\}$ we query oracles whether
    $\subsums{Z_1}{t}$ and $\subsums{Z_2}{t}$
    contain some element in the interval,
    having in mind that the answer is approximate.
    The number of intervals is $\Oh(\frac{1}{\eps})$.
    
    We store the answers in arrays
    $S_1$ and $S_2$, namely $S_j[i] = 1$ if the oracle for $\subsums{Z_j}{t}$  answers yes for interval $(ip, (i+1)p]$.
    \begin{displaymath}
        S'_1[i] = \begin{cases}
            1 & \text{if the oracle for } \;\; (ip, (i+1)p] \sqcap S_1 \;\; \text{or} \;\; i=0\\
            0 & \text{otherwise}
        \end{cases}
    \end{displaymath}
     Then we perform a fast convolution on $S_1,S_2$ with FFT.
    
    If $x \in \subsums{Z_1 \cup Z_2,t} \cap (kp, (k+1)p]$, then
    there is some $x_1 \in \subsums{Z_1}{t}$ and $x_2 \in\subsums{Z_2}{t}$ such that
        $x = x_1 + x_2$.
        We have  $(S_1 \oplus_{\text{FFT}} S_2)[k] =
    \sum_{i=0}^k S_1[i] \cdot S_2[k-i]$
and thus $(S_1 \oplus_{\text{FFT}}
    S_2)[k']$ is nonzero for $k'=k$ or $k'=k+1$.
    This defines the rule for the new oracle.
    The additive error of the oracle gets doubled with the summation.
    On the other hand, if one of these fields is nonzero, then there are corresponding indices $i_1, i_2$ summing to $k$ or $k+1$.
    The second condition from Definition~    \ref{mem-oracle} allows the corresponding value $x_1$ to lie within one of the intervals with indices $i_1 - 1,i_1$, or $i_1+1$ and likewise for $x_2$.
 Therefore, the additive error is $\Oh(p) = \Oh(\eps t)$.

    iff there is $i$ such that $(ip, (i+1)p] \sqcap S_1
    \cup \{0\}$ and $((k-i)p, (k-i+1)p] \sqcap S_2 \cup \{0\}$.

    Now, we promised only oracle output to our array. When a query comes, we
    scale down the query interval, then we check whether any of adjacent
    interval in our structure is nonzero (we lose a constant factor of
    $\Oh(\eps)$ accuracy here) and output \textbf{yes} if we found it and
    \textbf{no} otherwise.
    
    Moreover, with additional polylogarithmic factors we can also retrieve the
    solution. The idea is similar to backtracking from~\cite{kellerer-subsetsum}. Namely, the fast convolution algorithm
    can compute the table of witnessing indexes (of
    only one).
    We store a witnessing index if there is
    solution and $-1$ otherwise. Then we ask the oracles of $\subsums{Z_1}{t}$ and $\subsums{Z_2}{t}$ for a
    solution with a proper indexes and return the combination of those.

\end{proof}

\subsection{From Exact Solution to $\eps$-close Instance}

In Section~\ref{sec:techniques} we presented an overall approach of rounding elements 
and explained why it gives us the weak approximation of \ssum. Here we will
focus on formally proving these claims.

In our subroutines, we round down the items, execute the exact algorithm on the rounded instance, and retrieve the solution. We want to argue, that in the end we
lose only an additive factor of $\pm \eps t$.
We presented a sketch of this reasoning in Fact~\ref{thm:weak-example}.
For our purposes we will describe the procedure in the case,
when the number of items in any solution is bounded by $k$ (i.e., we are interested only in $\subsums{Z}{t}_k$). We can always assume $k \le n$.

\begin{lemma}
    \label{exact-approx}
    Given an exact algorithm that outputs the set $\subsums{Z}{t}_{k}$ and works in time
    $T(n,t)$, where $n = |Z|$, we can construct an \oracle of set
    $\subsums{Z,t}_{k}$ in time $\Ot(n+T(n,k/\eps))$.

    If the exact algorithm retrieves solution in $\Ot(1)$ time, then so
    does the oracle.
\end{lemma}
\begin{proof}
For sake of legibility, we assume that we are interested in $\subsums{Z}{t}_{k-1}$ - this only allows us to write simpler formulas.
    Let $(z_i)$ denote the items. We perform rounding in the following way:

    \begin{displaymath}
        z'_i  =  \floor{\frac{kz_i}{\eps t}}, \;\;\;\;\;
        t'  =  \floor{\frac{kt}{ \eps t}} =  \floor{\frac{k}{\eps}}.
    \end{displaymath}

  We run the
    exact algorithm on the rounded instance $(Z',t')$. It takes time $T(n,t') =
    \Oh(T(n,k/\eps))$. This algorithm returns
    $\subsums{Z'}{t'}_{k-1}$, which we store in array $Q[1,t']$. We construct \oracle in array $Q'[1,t']$ as follows: we set $Q'[i] = 1$ iff $Q$ contains 1 in range $(i-2k,i+k]$.
    If we want to be to able retrieve a solution, we need to also remember a particular index $j(i) \in (i-2k,i+k]$ such that $Q[j(i)] = 1$.
    Such a data structure can be constructed in a linear time with a help of a queue.
    Given a query $q$, the oracle returns
    $Q'[q']$, where  $q' = \floor{\frac{k q }{\eps t}}$.
    It remains to prove that  Definition~\ref{mem-oracle}
 is satisfied.

Let $I\subseteq Z$ be a set of at most $k-1$ items and $I'$ be the
    set of their counterparts after rounding.
    Since for all $z_i \in Z$ it holds
    
        \begin{displaymath}
         {\frac{k z_i}{\eps t}} - 1 \le z'_i \le {\frac{k z_i }{\eps t}},
    \end{displaymath}    
    we obtain
    
  \begin{equation}
      \label{equation:star}
        {\frac{k\cdot\Sigma(I)}{\eps t}} - k + 1 \le \Sigma(I') \le {\frac{k\cdot\Sigma(I')}{\eps t}}.\quad
    \end{equation}  
Therefore, if $\Sigma(I) \in [q - \eps t, q + \eps t]$, then 

  \begin{align*}
     &  \frac{k q}{\eps t} -2k + 1 =  {\frac{k\cdot (q - \eps t)}{\eps t}} - k + 1 \le \Sigma(I'), \\
     & \Sigma(I') \le {\frac{k\cdot (q + \eps t)}{\eps t}} = \frac{k q}{\eps t}+ k,
    \end{align*}
and $\Sigma(I') \in (q' - 2k, q' + k]$, because $\Sigma(I')$ is integer.
    On the other hand, we can invert relation (\ref{equation:star}) to obtain

  \begin{align*}
 \frac{\eps t}{k}\cdot \Sigma(I')  \le \Sigma(I)  \le  \frac{\eps t}{k}\cdot \left(\Sigma(I') + k- 1\right).
    \end{align*}
To satisfy the second condition we assume $\Sigma(I') \in (q' - 2k, q' + k]$ and check that
  \begin{align*}
     q - 2\eps t &= \frac{\eps t}{k}\cdot \left(\frac{k q}{\eps t} -2k  \right) \le \frac{\eps t}{k}\cdot (q'-2k+1)  \le \Sigma(I), \\
   &  \Sigma(I)  \le \frac{\eps t}{k}\cdot (q'+2k - 1) \le q + 2\eps t,
    \end{align*}
what finishes the proof.
\end{proof}
  We apply Lemma~\ref{rounding-numbers-d} with
    $\{z_1,\ldots, z_k\} = Y$, $q = t/2$ and $k$ and $\eps$ as in the statement.
   It guarantees that:

    \begin{displaymath}
        (1-\eps) \Sigma(Y) \le \frac{\eps t}{2k} \Sigma(Y')
        .
    \end{displaymath}

    And finally, $(1-\eps) \Sigma(Y) \ge \Sigma(Y) - \eps t$ (because we are only
    interested in solutions smaller than $t$). So if $Y$ is a optimal solution,
    then an exact algorithm after rounding would return something greater or
    equal $\Sigma(Y) - \eps t$.

    Conversely, it can turn out that an exact algorithm would find something with
    a sum greater than $q$ (this is where we can violate the hard constraint).
    We need to bound it as well (because the definition of \oracle requires that). Note,
    that analogous argument proves it. Namely, the solution can consist of at most
    $k$ items and each of them lose only $\Oh(\eps t/k)$. Moreover, exact oracle gave
    us only the solution that its rounded version sums up to exactly $t'$.
    Formally, we prove it again with Lemma~\ref{rounding-numbers-d} with the
    same parameters as before. By dividing both sides by $(1-\eps)$ we know that:

    \begin{displaymath}
        \sum^k_{i=1} \floor{\frac{2 k z_i}{t \eps}} = \floor{k}\cdot\eps.
    \end{displaymath}

    Once again, we can use Lemma~\ref{rounding-numbers-d} with the same
    parameters (we divided both sides by $(1-\eps) > 0$):

    \begin{displaymath}
        \sum_{i=1}^k x_i \le 
        \left(\frac{1}{1-\eps}\right) \frac{\eps t}{2k}\sum_{i=1}^k \floor{\frac{2 k z_i}{t \eps}}.
    \end{displaymath}
The right side satisfies:

    \begin{displaymath}
        \left(\frac{1}{1-\eps}\right) \frac{\eps t}{2k}\sum_{i=1}^k \floor{\frac{2 k z_i}{t \eps}}
        = \left(\frac{1}{1-\eps}\right) \frac{\eps t}{2k}
        \floor{\frac{2k}{\eps}} \le \frac{1}{1-\eps} t < (1+2\eps) t.
    \end{displaymath}

    The constant before $\eps$ does not change much since we only need
    \oracleO (we can always rescale the approximation
        factor by setting $\eps' = \eps/2$ at the beginning).    

The main obstacle with returning a solution that obeys the capacity constraint
comes from the above lemma. If we could provide a reduction from an exact algorithm without widening the interval $[q-\eps t, q + \eps t]$,
this would automatically entail a strong approximation for \ssum.
This seems unlikely due to conditional hardness result for a strong subquadratic approximation for \ssum~\cite{bringmann-com}.

At the end, we  need to prove, that an \oracle gives us the
correct solution for \weak \ssum. 

\begin{lemma}
    \label{from-oracle-to-ssum}
    Given an $(\frac \eps 6,t$)-membership-oracle of $\subsums{Z}{t}$, we can read the answer to the \weak \ssum in time $\Ot(\frac{1}{\eps})$.
\end{lemma}
\begin{proof}
We query the oracle for $q = i \cdot \frac{\eps t}{6}$ for $i=0,\dots,\frac 6 \eps$.
Each query takes time $\Ot(1)$ and if the interval $[q- \frac{\eps t}{6}, q+ \frac{\eps t}{6}]$ contains an $x \in \subsums{Z}{t}$, then the oracle returns an element within $[x- \frac{\eps t}{2}, x+  \frac{\eps t}{2}]$.
If $\text{OPT} < (1-\frac\eps 2)t$, then the oracle will return a witness within $(\text{OPT}- \frac{\eps t}{2}, \text{OPT}]$.
Otherwise the witness might belong to $(t, (1+\frac \eps 2)t]$.

By taking advantage of Lemma~\ref{lem:large-opt}, we can assume that $\text{OPT} \ge t/2$,
therefore the relative error gets
bounded with respect to $\text{OPT}$.
\end{proof}
\section{The \weak algorithm for \ssum}
\label{fptas-ssum}

\subsection{Large Items}

We will use 
Theorem~\ref{thm:bringmann}  to compute $\subsums{\zzlarge}{t}$ on a large
instance. On that instance, this algorithm is more efficient than  \kellerer
because one can round items less aggressively.

\begin{lemma}[Algorithm for Large Items]
    \label{lem:large-algorithm}
    Given a large instance $(\zzlarge,t)$ of \ssum (i.e., all items are greater
    than $\gamma t$), we can construct an \oracle of $\subsums{\zzlarge}{t}$ in
    randomized $\Ot(n+\frac{1}{\gamma \eps})$ time with a constant probability
    of success.
\end{lemma}

\begin{proof}
    We use \bringmann, namely Corollary~\ref{cor:bringmann-all}, that solves
    the \ssum problem exactly. 
    Since all elements are
    greater than $\gamma t$, any subset that sums up to
    at most $t$ must contain at most $\frac{1}{\gamma}$ items. The parameter $k$
    in Lemma~\ref{exact-approx} is an upper bound on number of elements in the
    solution, hence we set  $k = \frac{1}{\gamma}$. The \bringmann runs in time $\Ot(n+t)$ and
    Lemma~\ref{exact-approx} guarantees that we can build an \oracle
    in time $\Ot(n + k/\eps) = \Ot(n + 1/(\gamma \eps))$, which is what we needed.
    
\end{proof}

\subsection{Small Items}
\label{smallitems}

Now we need an algorithm that solves the problem for small items. As
mentioned in Section~\ref{sec:techniques} we will consider two cases depending
on the density of instance.
The initial \ssum instance consists of $n$ elements. The $m$ is the
number of elements in the small instance and let $m' = \Oh(m\log n)$
be as in Lemma~\ref{lem:partition-small-large}.
For now we will assume, that the set of
elements is $(\eps t/m')$-distinct (we will deal with multiplicities 2 in
Lemma~\ref{lem:multiplicity-small}).

Let $q = \eps t/m'$
be the rounding parameter (the value by which we divide) and $\ell = \gamma m' / \eps =  \Oh(\frac{\gamma m
\log{n}}{\eps})$ be the upper bound on
items' sizes in the small instance after rounding. Parameter $L = \Oh(\Sigma(S)\cdot
\frac{l}{m^2})$ describes the boundaries of Theorem~\ref{galil-algorithm}. We deliberately use $\Oh$ notation to hide constant
factors (note that \galil requires that $m > 1000\cdot\sqrt{l}\log{l}$).

\begin{lemma}[Small items and $m^2 < \ell\log^2{\ell}$]
    \label{small-sparse}
    Suppose we are given an instance $(\zzsmall,t)$ of \ssum (i.e., all items are smaller than
    $\gamma t$) with size satisfying $m^2 < \ell \log^2{\ell}$.
    Then we can compute
    \oracle of $\subsums{\zzsmall}{t}$ in randomized $\Ot(m+\frac{\gamma}{\eps^2})$ time.
\end{lemma}

\begin{proof}
    In here we need to deal with the case, where \textit{small} instance is sparse.
    So just as in the proof of Lemma~\ref{lem:large-algorithm}, we can use \bringmann.

    We will use the reduction from exact to \weak algorithm for \ssum
    from Lemma~\ref{exact-approx}. We set $m$ as the maximal number of items in the
    solution,
    as there are at most $m$ small items. Recall that
    $\ell$ is $\Ot(m\gamma/\eps)$. This gives us  $m^2 = \Ot(\ell) = \Ot(\frac{m
    \gamma}{\eps})$. After dividing both sides by $m$ we obtain $m =
    \Ot(\frac{\gamma}{\eps})$.

    Combining Corollary~\ref{cor:bringmann-all} and Lemma~\ref{exact-approx} allows us to construct an
    \oracle in $\Ot(m+T(m, m/\eps)) =\Ot(m + \frac{\gamma}{\eps^2})$ randomized time.
\end{proof}

Now we have to handle the harder $m^2 \ge \ell\log^2{\ell}$ case. In this situation we
again consider two cases. The \galil allows only to ask queries
in the range $(L,\,\Sigma(S)-L)$ where $L = \Oh(\Sigma(S)\cdot
\frac{l}{m^2})$. In the next lemma we take care of ranges $[0,
L]$ and $[\Sigma(S) -L, \Sigma(S)]$.
We focus on the range $[0,L]$, because the sums within
$[\Sigma(S)-L, \Sigma(S)]$ are symmetric to $[0,L]$.

\begin{lemma}[Small items, range $(0,L)$]
    \label{small-range-bringmann}
    Given an  instance $(\zzsmall,t)$ of \ssum, such that $|\zzsmall| = m$ and the items' sizes are at most $\gamma t$, we can compute an \oracle for $\subsums{\zzsmall, L}$ in time $\Ot(m+\frac{m\gamma^2}{\eps^2})$.
\end{lemma}

\begin{proof}
    We round down items with  rounding parameter $q = \eps t/m' = \Omega(\frac{\eps t}{m \log{n}})$ and denote
    the set of rounded items as $\zzsmall'$.
    After scaling down we have $L' = \Sigma(\zzsmall')\cdot \frac{c \ell}{m^2}$ (note that we only replace $\Sigma(\zzsmall)$ with $\Sigma(\zzsmall')$ and $\ell$ remains the same).
    Recall that $\ell = \Oh(\frac{m
\gamma \log{n}}{\eps})$.

    The total sum of items in $\zzsmall'$ is smaller or equal to $\ell m$
    (because there are $m$ elements of size at most~$\ell$). Hence ,$L'=
    \Oh(\ell^2/m) = \Oh(\frac{\gamma^2 m \log^2{n}}{\eps^2})$.
    Therefore \bringmann runs in time $\Ot(m + L') = \Ot(m+\frac{m\gamma^2}{\eps^2})$. Combining it with the analysis of the Lemma~\ref{exact-approx}
    gives us an \oracle for $\subsums{\zzsmall, L}$.
\end{proof}

\subsection{Applying Additive Combinatorics}

Before we proceed forward, we need to present the full theorem
of \citet[Theorem 6.1]{galil} (in Section~\ref{galil-box} we presented
only a short version to keep it free from technicalities). We need a full running time
complexity (with dependence on $\ell,m,\Sigma(S)$). 
We copied it in here with a slight
change of notation (e.g., \cite{galil} use $S_A$ but we use
notation from \cite{koiliaris-soda} paper of $\Sigma(A)$).

\begin{theorem}[Theorem 6.1 from \cite{galil}]
    Let $A$ be a set of $m$ different numbers in interval $(0,\ell]$ such that 
    \begin{displaymath}
        m > 1000 \cdot\ell^{0.5} \log_2{\ell};
    \end{displaymath}
    then we can build in $\Oh\left(m+\left( (\ell/m)\log{l} \right)^2 +
    \frac{\Sigma(A) }{m^2} \ell^{0.5} \log^2{\ell} \right)$
    preprocessing time a structure which allows us to solve the \ssum problem
    for any given integer $N$ in the interval $(L, \Sigma(A) - L)$. Solving
    means finding a subset $B \subseteq A$, such that $\Sigma(B) \le N$ and
    there is no subset $C \subseteq A$ such that $\Sigma(B) < \Sigma(C) \le N$.
    An optimal subset $B$ is build in $\Oh(\log{\ell})$ time per target number and
    is listed in time $\Oh(|B|)$. For finding the optimal sum $\Sigma(B)$ only,
    the preprocessing time is $\Oh\left(m + \left( (\ell/m)\log{\ell} \right)^2
    \right)$ and only constant time is needed per target number.
\end{theorem}

In \cite{galil} authors defined $L := \frac{100\cdot\Sigma(A) \ell^{0.5}
\log_2{\ell}}{m}$, however in the next \cite{galil-icalp} the authors improved
it to $L := \Oh(\Sigma(A) \frac{\ell}{m^2})$ without any damage on running time~\cite{galil-comunication}. For both
of these possible choices of L we  obtain a subquadratic algorithm. 
We will use the improved version~\cite{galil-icalp} because it provides a
better running time.

\begin{lemma}[Small items, range $(L,\,\Sigma(S)-L)$]
    \label{lem:small-algorithm}
    Given a small instance $(\zzsmall, t)$  of \ssum (i.e., all items are smaller
    than $\gamma t$) such that $\zzsmall$ is $(\eps t/m')$-distinct (where $m' = \Oh(m \log{n})$), we can compute an
    \oracle of $\subsums{\zzsmall}{t} \cap (L,\, \Sigma(\zzsmall)-L)$
    in time
    $\Ot(n + \left( \frac{\gamma}{\eps} \right)^2 + \frac{\gamma}{\eps}\cdot \left(
    \frac{\gamma n}{\eps} \right)^{0.5})$.

\end{lemma}

\begin{proof}
We round items to multiplicities of $q
    = \eps t / m'$. Precisely:

    \begin{displaymath}
        z_i' = \floor{\frac{z_i}{q}}, \;\;\;\;\;
        t'   = \floor{\frac{t}{q}} = \floor{\frac{m'}{\eps}}.
    \end{displaymath}

    We know that $z_i < \gamma t$. Therefore

    \begin{displaymath}
        z_i' \le \frac{z_i}{q} < \frac{\gamma t}{q} = \frac{\gamma m'}{\eps} = \ell.
    \end{displaymath}

    By the same
    inequalities as in the proof of Lemma~\ref{exact-approx} we know that if we compute
    $\subsums{\zzsmall',t'}$ and multiply all results by $q$, we obtain an \oracle
for $\subsums{\zzsmall,t}$.
    
    \paragraph{Checking conditions of the algorithm}

    Now we will check that we satisfy all assumptions of \galil on the rounded
    instance $\zzsmall'$. First note that $m^2 <
    \ell\log^2{\ell}$, $\ell$ is the upper bound on the items' sizes in $\zzsmall'$, and we know
    that all items in $\zzsmall'$ are distinct because we assumed that
    $\zzsmall$ is $(\eps t/m')$-distinct.

    \paragraph{Preprocessing}

    Next \galil constructs a data structure on the set of rounded items
    $\zzsmall'$. The preprocessing of \galil requires 
    $$\Oh\left( m+ \left( \ell/m
\log{\ell} \right)^2 + \frac{\Sigma(\zzsmall')}{m^2}\ell^{0.5} \log^2{\ell}
    \right)$$
    time. If we put it in terms of $m,\eps,t$ and hide polylogarithmic factors we see that
    preprocessing runs in:

    \begin{displaymath}
        \Ot\left(m + \left( \frac{\gamma}{\eps} \right)^2 + \frac{\gamma}{\eps} \left(
        \frac{\gamma m}{\eps} \right)^{0.5}\right)
    \end{displaymath}
    because $\Sigma(\zzsmall') \le \ell m$.

    \paragraph{Queries}

    With this data structure we need to compute a set $\eps$-close to
    $\subsums{\zzsmall,t} \cap (L,\,\Sigma(\zzsmall')-L)$.
   After scaling down we have $L'= \Ot\left(\Sigma(\zzsmall')\cdot\frac{\ell}{m^2}\right) = \Ot(\frac{\ell^2}{m}) = \Ot(\frac{\gamma^2
    m^2}{m\eps^2}) = \Ot(\frac{m\gamma^2}{\eps^2})$.

    Naively, one could run queries for all elements in range
    $(L',\Sigma(\zzsmall')-L')$ and check if there is a subset of $\zzsmall'$ that
    sums up to the query value. However this is too expensive.
    In order to deal with this issue, we take advantage of the fact that each query returns the closest set whose sum is smaller or equal to the query value.

    Since we have rounded down items with $q = \frac{\eps t}{m'}$, we only need to ask $\frac{\eps t}{q}
    = \Ot(m)$ queries in order to learn sufficient information.
    The queries will reveal if $\zzsmall$ contains at least one element in
    each range $[i\eps t, (i+1)\eps t)$, what matches the definition of the
    \oracle.

    \paragraph{Retrieving the solution}

    \galil can retrieve the solution in time $\Oh(\log{\ell})$. 

    This finalizes the construction of the \oracle. The running time is dominated
    by the preprocessing time.

\end{proof}

\subsection{Combining the Algorithms}

Now we will combine the algorithms for small items.

\begin{lemma}[Small Items]
    \label{small-algorithm}
    Given a $(\zzsmall,t)$ instance of \ssum (i.e., all elements in $\zzsmall$
    are smaller than $\gamma t$), such that the set $\zzsmall$ is $(\eps t/m)$-distinct,
    we can compute an
    \oracle of $\subsums{\zzsmall}{t}$
    in time
    $\Ot(m + \frac{\gamma}{\eps^2} + \frac{m\gamma^2}{\eps^2} )$ with high probability.
\end{lemma}

\begin{proof}
    We will combine two cases:
    
    \paragraph{Case When $m^2 < \ell\log^2{\ell}$:}

    In such case we use Lemma~\ref{small-sparse} that works in 
    $\Ot(m + \frac{\gamma}{\eps^2})$ time.

    \paragraph{Case When $m^2 \ge \ell\log^2{\ell}$:}
    First we take advantage of Lemma~\ref{lem:small-algorithm}.
    This gives us an \oracle that answers queries within set
    $\subsums{\zzsmall}{t} \cap (L, \Sigma(\zzsmall)-L)$. It requires 
    $\Ot(nm+ \left( \frac{\gamma}{\eps} \right)^2 + \frac{\gamma}{\eps}\cdot \left(
    \frac{\gamma m}{\eps} \right)^{0.5})$ time.

    We combine it (using Lemma~\ref{FFT}) with the \oracle that gives us
    answers to a set $\subsums{\zzsmall}{t} \cap [0,L]$ from
    Lemma~\ref{small-range-bringmann}. This oracle can be constructed in time
    $\Ot(m+\frac{m\gamma^2}{\eps^2})$. The oracle for interval set $[\Sigma(\zzsmall) - L,
    \Sigma(\zzsmall)]$ is obtained by symmetry.

    \paragraph{Running Time:}

    The running time of merging the solutions from Lemma~\ref{FFT} is
    $\Ot(1/\eps)$ which is suppressed by the running time of
    Lemma~\ref{small-range-bringmann} and Lemma~\ref{lem:small-algorithm}.
    Factor $\Ot( \left( \frac{\gamma}{\eps} \right)^2)$ is suppressed by
    $\Ot\left( \left( \frac{m \gamma^2}{\eps^2} \right)\right)$.

    Term $\frac{\gamma}{\eps}\cdot \left( \frac{\gamma m}{\eps}
    \right)^{0.5}$ is also suppressed by $\Ot(\frac{m \gamma^2}{\eps^2})$. The
    algorithm is randomized because Lemma~\ref{small-range-bringmann} is
    randomized.

\end{proof}

The Lemma~\ref{lem:partition-small-large} allowed us to partition our instance into small
and large items. We additionally know that each interval of length $\eps t/m'$ contains
at most 2 items. However in the previous
proofs we assumed there can be only one such item, i.e., the set should be $(\eps t/m')$-distinct.

\begin{lemma}[From multiple to distinct items]
    \label{lem:multiplicity-small}
    Given an instance $(\zzsmall,t)$ of \ssum, where $|\zzsmall| = m$ and $zzsmall$ is $(\eps t/m',2)$-distinct for $m' = \Oh(m\log n)$, we can compute an \oracle for instance
    $(\zzsmall,t)$ in 
    $\Ot(n + \frac{\gamma}{\eps^2} + \frac{n\gamma^2}{\eps^2} )$ time with high probability.
\end{lemma}

\begin{proof}
    We divide the set $\zzsmall$ into two sets $\zzsmall^1$ and $\zzsmall^2$ such
    that $\zzsmall = \zzsmall^1 \cup \zzsmall^2$, the sets
    $\zzsmall^1$,$\zzsmall^2$ are disjoint, $(\eps t/m')$-distinct, and have size $\Omega(m)$.
    This can be done by
    sorting $\zzsmall$ and dividing items into odd-indexed and even-indexed . It takes $\Ot(m)$
    time.
    
    Next we use Lemma~\ref{small-algorithm} to compute an
    \oracle for $(\zzsmall^1,t)$ and $(\zzsmall^2,t)$,
    and merge them using Lemma~\ref{FFT}.
\end{proof}

Now we will combine the solutions for small and large items.

\begin{theorem}\label{lem:ssum-k}
    Let $0 < \gamma$ be a trade-off parameter (that depends on $n,\eps$).
    Given an $(Z,t)$ instance of \ssum, we can construct the \oracle
    of instance $\subsums{Z}{t}$ in
    $\Ot(n + \frac{1}{\gamma\eps} + \frac{\gamma}{\eps^2} + \frac{n\gamma^2}{\eps^2} )$ time with high probability.
\end{theorem}
\begin{proof}
    We start with Lemma~\ref{lem:partition-small-large}, that in $\Oh(n
    \log^2{n})$ time partitions the set into $\zzlarge$ and $\zzsmall$, such that
    $\zzsmall$ is $(\frac{\eps t}{m\log n},2)$-distinct, where $m=|\zzsmall|$.
    To deal with small items, we use Lemma~\ref{lem:multiplicity-small}. The algorithm for small items returns
    an \oracle of $\subsums{\zzsmall}{t}$.  
    For large items we can use Lemma~\ref{lem:large-algorithm}. It also returns an
    \oracle of $\subsums{\zzlarge}{t}$.

    Finally, we use Lemma~\ref{FFT} to merge these oracles in time $\Ot(1/\eps)$.
All the subroutines run with a constant probability of success.
\end{proof}

Finally, we have combined all the pieces and we can get a faster algorithm for
\weak for \ssum.

\begin{corollary}[\ssum with tradeoff]
    \label{weak-tradeoff}
    There is a randomized \weak algorithm for \ssum running in $\Ot(n +
    \frac{1}{\gamma\eps} + \frac{\gamma}{\eps^2} + \frac{n\gamma^2}{\eps^2} )$ time with high probability for any
     $\gamma(n,\eps) > 0$.
\end{corollary}
\begin{proof}
    It follows from Lemma~\ref{from-oracle-to-ssum} and Theorem~\ref{lem:ssum-k}.
\end{proof}

The \weak \ssum gives us the approximation for \partition via
Corollary~\ref{partition-reduction}.

\begin{corollary}[\partition with trade-off]
    There is a randomized $(1-\eps)$-approximation algorithm for \partition running in $\Ot(n +
    \frac{1}{\gamma\eps} + \frac{\gamma}{\eps^2} + \frac{n\gamma^2}{\eps^2} )$ time with high probability for any
     $\gamma(n,\eps) > 0$.
\end{corollary}

To get running time of form $\Ot(n+1/\eps^c)$ and prove our main
result we need to reduce the number of items from $n$ to $\Ot(1/\eps)$ and choose
the optimal $\gamma$.

\begin{theorem}[Weak apx for \ssum]
    \label{weak-linear-ssum}
    There is a randomized \weak algorithm for \ssum running in $\Otilde\big(n + \eps^{-\frac{5}{3}}\big)$ time.
\end{theorem}
\begin{proof}
    We apply Lemma~\ref{lem:reduce-items} to ensure that the number of items is $\Otilde\big(\frac{1}{\eps}\big)$ and work with an $\Oh(\eps)$-close instance.
    Then we take advantage of Corollary~\ref{weak-tradeoff} with $\gamma = \eps^{\frac{2}{3}}$.
\end{proof}

Analogously for \partition we get that:
\begin{theorem}[Apx for \partition]
    There is a randomized $(1-\eps)$-approximation algorithm for \partition running in $\Otilde\big(n + \eps^{-\frac{5}{3}}\big)$ time.
\end{theorem}
\section{Approximate \minconv}
\label{approximate-minconv}

\defproblem{Approximate \minconv}
{Sequences \sequence{A}{0}{n-1}, \sequence{B}{0}{n-1} of positive integers and
approximation parameter $0 < \eps < 1$}
{Let $\text{OPT}[k] = \min_{0 \le i \le k} (A[i]+B[k-i])$ be the \minconv of
$A$ and $B$. Find a sequence \sequence{C}{0}{n-1} such that $\forall_i\, \text{OPT}[i] \le C[i] \le (1+\eps)\text{OPT}[i]$}

\citet{tree-sparsity} described a $(1+\eps)$-approximation algorithm for
\minconv, that runs deterministically in time $\Oh(\frac{n}{\epsilon^2} \log{n}
\log^2{W})$. In their paper~\cite{tree-sparsity} it is used as a building block
to show a near-linear time approximation algorithm for \sparsity. With the
approximation algorithm for \minconv, they managed to solve \sparsity
approximately in $\Ot(\frac{n}{\eps^2})$ time, which in practical applications
may be faster than solving this problem exactly in time $\Ot(n^2)$.

We begin with explaining its connection with the \subsetsum
problem.
A natural generalization of \subsetsum is \knapsack. In this scheme each item has
\textit{value} and \textit{weight} and our task is to pack items into the
knapsack of capacity $C$, so that their cumulative weight does not exceed
capacity and in the same time we want to maximize their total value. 
In the special case when all weights and values are equal  we obtain
the \subsetsum problem.

To certify the existence of a subset with a given sum, we have used the fast
convolution using FFT as
subroutine. If we want to generalize it and capture maximal value subset of
items of a given weight, we require \maxconv,
which is computationally equivalent to \minconv.

\citet{icalp2017} exploited this idea to show subquadratic equivalence between
exact \minconv, \knapsack, and other problems.
Here we focus on the approximate setting.
From~\cite{icalp2017} it follows that the $\Ot(n+\frac{1}{\eps^{1.99}})$
approximation algorithm for \uknapsack is unlikely. This lower bounds proves the
optimality of \citet{fptas-uknapsack} $\Ot(n+\frac{1}{\eps^2})$ FPTAS for
\uknapsack. The current best FPTAS for \knapsack is burdened with
time complexity of $\Ot(n+1/\eps^{12/5})$~\citet{chan-knapsack}. We hope, that our
approximation schemes are a step towards a faster FPTAS for \knapsack.

In this section we improve upon the $\Ot(n/\eps^2)$ approximation algorithm.
for \minconv.
Similar techniques have been
exploited to obtain the $\Ot(n^\omega/\eps)$-time approximation for APSP~\cite{approx-apsp}
and they have found use in the approximate pattern matching over
$l_\infty$~\cite{approx-pattern-matching}. The basic idea is to propose a fast
exact algorithm depending on $W$ (upper bound on the weights) and
apply it after rounding weights into smaller space.
Our result also applies to \maxconv.

%
%
%
%
%
%

\subsection{Exact $\Ot(nW)$ algorithm}
\label{exact-minconv}

The \minconv admits a brute force $\Oh(n^2)$-algorithm. From the other hand, when all
values in sequences are binary, then applying FFT and performing convolution
yields an $\Oh(n\log{n})$-algorithm.
Our exact $\Ot(nW)$ algorithm is an attempt to capture this
trade-off. Note, that this algorithm is worse than a brute force whenever $W>n$ which
is often the case. However, this
algorithm turns out useful for approximation.

\begin{lemma}
    \label{fast-exact-minconv}
    The \minconv \emph{[$(\max,+)$-convolution]} problem can be solved
    deterministically in $\Oh(nW\log{(nW)})$ time
    and $\Oh(nW)$ space.
\end{lemma}

\begin{proof}
    Given sequences \sequence{A}{0}{n-1} and \sequence{B}{0}{n-1} with
    values at most $W$, we transform them into binary sequences of length $2nW$.
    We encode every number in the natural unary
    manner. For $0 \le i < n,\, 1 \le k \le W$ we define:

    \begin{displaymath}
        \tilde{a}[2Wi+k] = \begin{cases} 
            0 & \text{if} \; A[i] \ne k \\
            1 & \text{if} \; A[i] = k
        \end{cases}
    \end{displaymath}      
    and similarly we define sequence $\tilde{B}$. 
    For example, sequence $(2,3,1)$ with $W=3$ gets encoded as
    $010`000`001`000`100`000$ (the separators ` are used to visually separate
    sections of length $W$).

	We compute convolution $\tilde{C} = \tilde{A} \oplus \tilde{B}$ using
    FFT in time $\Oh(nW\log{n}\log{W})$.
    Since $\tilde{C}[2Wi+k] = \sum_{\genfrac{}{}{0pt}{}{i_1 + i_2 = i}{k_1 + k_2 =k}} \tilde{A}[2Wi_1+k_1]\cdot\tilde{B}[2Wi_2+k_2]$,
    the first nonzero occurrence in the $i$-th block of length $2W$
    encodes the value of the $i$-th element of the requested $(\min,+)$-convolution.
    If we are interested in computing \maxconv, we should similarly
    seek for last nonzero value in each block.

    The time complexity is dominated by performing convolution with FFT.
    As the additional space we need $\Oh(nW)$ bits for the transformed sequences.
\end{proof}

\subsection{Approximating Algorithm}

We start with a lemma
inspired by~\cite[Lemma 5.1]{approx-apsp} and \cite[Lemma~1]{approx-pattern-matching}.

\begin{lemma}
    \label{numbers-rounding}
    For natural numbers $x,y$ and positive $q,\eps$ satisfying $q \le x+y$
    and $0 < \eps < 1$ it holds:
    \begin{eqnarray*}
        x+y &\le \Big(\Big\lceil \frac{2x}{q\epsilon} \Big\rceil +
        \Big\lceil \frac{2y}{q\epsilon} \Big\rceil\Big)\frac{q\epsilon}{2} &< (x+y)(1+\epsilon), \\
        (x+y)(1-\epsilon) &< \Big(\Big\lfloor \frac{2x}{q\epsilon} \Big\rfloor +
        \Big\lfloor \frac{2y}{q\epsilon} \Big\rfloor\Big)\frac{q\epsilon}{2} &\le x+y.
    \end{eqnarray*}
\end{lemma}
\begin{proof}
    The proof is a special case of Lemmas~\ref{knumbers-rounding}
    and~\ref{rounding-numbers-d} for $k=2$.
\end{proof}

\begin{lemma}
    \label{approx-reduction}
    Assume the \minconv \emph{[$(\max,+)$-convolution]} can be solved exactly in time $T(n,W)$. Then we can
    approximate \minconv \emph{[$(\max,+)$-convolution]} in time $\Oh((T(n,\frac{4}{\eps}) + n)\log{W})$.
\end{lemma}
\begin{algorithm}
    \caption{$\textsc{ApproximateMinConv}(A,B)$. We use a simplified notation to
    transform all elements in the sequences $A[i]$ and $B[i]$.}
	\label{alg:approx-minconv}
\begin{algorithmic}[1]
    \State $\text{Output}[i] = \infty$
    \For {$l=2\lceil\log{W}\rceil,\ldots,0$}
      \State $q := 2^l$ 
      \State $A'[i] = \lceil \frac{2A[i]}{q \eps} \rceil$ 
      \If {$A'[i] > \ceil{4/\eps}$} 
         \State $A'[i] = \infty$
      \EndIf
      \State $B'[i] = \lceil \frac{2B[i]}{ q \eps} \rceil$ 
      \If {$B'[i] > \ceil{4/\eps}$}
         \State $B'[i] = \infty$
      \EndIf
      \State $V = \text{runExact}(A',B')$ 
      \If { $V[i] < \infty$ }
          \State $\text{Output}[i] = V[i] \cdot \frac{q \eps}{2}$ 
      \EndIf
    \EndFor
    \State \Return $\text{Output}[0,\dots,n-1]$
\end{algorithmic}
\end{algorithm}
\begin{proof}
The idea is based on~\cite[Section 6.2]{approx-pattern-matching}.
We focus on the variant with \minconv, however the proofs works alike
for \maxconv.

We iterate the \emph{precision parameter} $q$ through $2W,W,\ldots,4,2,1$.
In each iteration we apply the transform from Lemma~\ref{numbers-rounding}
($x \rightarrow \big\lceil \frac{2x}{q\epsilon} \big\rceil$) to all elements in $A,B$,
we set $\infty$ for each value exceeding $\ceil{\frac{4}{\eps}}$, and launch
the exact algorithm on such input.
We multiply all finite elements in the returned array by $\frac{q\epsilon}{2}$
and store them in the output array $C$,
possibly overwriting some elements.

Assume the correct value of $C[k]$ equals $A[i] + B[k-i]$.
For some iteration we get the precision parameter $q$ such that
$q \le C[k] < 2q$.
The rounded numbers $\Big\lceil \frac{2A[i]}{q\epsilon} \Big\rceil,\,
\Big\lceil \frac{2B[k-i]}{q\epsilon} \Big\rceil$ are at most $\ceil{\frac{4}{\eps}}$,
so we will update the $k$-th index in the output array.
On the other hand, the assumption of Lemma~\ref{numbers-rounding} is satisfied,
therefore the generated value lies between $C[k]$ and $C[k](1+\eps)$.
In the following iterations, we will still have $q \le C[k]$,
therefore any further updates to the $k$-th index will remain valid. 
    
The algorithm performs $\Oh(\log{W})$ iterations and in each
step we run the exact algorithm in time $T(n,\frac{4}{\epsilon})$, thanks
to the pruning procedure. Transforming the
sequences takes $\Oh(n)$ time in each step.
\end{proof}


\begin{theorem}[Apx for $(\min/\max,+)$-conv]
    \label{approx-minconv}
    There is a deterministic algorithm for $(1+\eps)$-approximate \minconv \emph{[$(\max,+)$-convolution]} running in
    $\Oh\left(\frac{n}{\eps}\log{(\frac{n}{\eps})}\log{W}\right)$ time.
\end{theorem}
\begin{proof}
    From Lemma~\ref{fast-exact-minconv} the running time of exact algorithm is $T(n,W) =
    \Oh(nW\log{n}\log{W})$. This quantity dominates the additive term $\Oh(n\log{W})$.
    Hence by replacing each $W$ with $1/\eps$
    we get the claimed running time.
\end{proof}

\section{Tree Sparsity}\label{sec:tree-sparsity}

The \sparsity problem has been stated as follows: given a node-weighted binary tree and an integer $k$, find a rooted subtree of size $k$ with the maximal weight.
Its approximation version comes with two flavors:
as a \emph{head} approximation where we are supposed to maximize the weight of the solution, and as a \emph{tail} approximation where we minimize the total weight of nodes that do not belong to the solution.
Note that a constant approximation for one of the variants does not necessarily yield a constant approximation for the other one.
\citet{tree-sparsity} proposed
an $\Oh\big(\frac{n}{\eps^2}\cdot\log^{12}{n}\cdot\log^2{W}\big)$ running time for $(1-\eps)$-head approximation,
and an $\Oh\big(\frac{n}{\eps^3}\cdot\log^9{n}\cdot\log^3{W}\big)$ running time for $(1+\eps)$-tail approximation.

In this section we improve the running times for both variants relying on the $\Otilde\big(\frac{n}{\eps}\big)$ algorithm for approximating $(\min,+)$ and $(\max,+)$ convolutions.
Our construction is based on the approach by~\citet{icalp2017}
which also results in a simpler analysis than for the previously known approximation schema~\cite{tree-sparsity}.
In particular, a single proof suffices to cover both head and tail variants.

The following theorem, combined with our approximation for \minconv  yields an $\Oh\big(\frac{n}{\eps}\cdot\log(n/\eps)\cdot\log^3{n}\cdot\log{W}\big)$-time algorithm that computes the maximal weights of rooted subtrees for each size $k=1,\dots,n$ with a relative error at most $\eps$ in both head and tail variant.

\begin{theorem}

    If $(1+\eps)$-approximate \minconv can be solved in time $T(n,W,\eps)$, then
    $(1+\eps)$-approximate \sparsity can be solved in time $\Oh\big(\big(n +
    T(n,W,\eps/\log^2{n})\big) \log{n}\big)$.

\end{theorem}
\begin{proof}
We exploit the heavy-light decomposition introduced
by~\citet{heavy-light}.
This technique has been utilized by~\citet{tree-sparsity} in their work on \sparsity approximation and later by~\citet{icalp2017} in order to show a subquadratic equivalence between \sparsity and \minconv.

We construct a \emph{spine} with a \emph{head} $s_1$ at the root of the tree.
We define $s_{i+1}$ to be the child of $s_i$ with the larger subtree (in case of draw we choose any child)
and the last node in the spine is a leaf.
The remaining children of nodes $s_i$ become heads for analogous spines so the whole tree gets covered.
Observe that every path from a leaf to the root intersects at most $\log{n}$ spines because
each spine transition doubles the subtree size.

At first we express the head variant in the convolutional paradigm.
For a node $v$ with a subtree of size $m$ we define the sparsity vector $(x^v[0], x^v[1], \dots, x^v[m])$ of weights of the heaviest subtrees rooted at $v$ with fixed sizes.
This vector equals the \maxconv of the sparsity vectors for the children of $v$.
We are going to compute sparsity vectors for all heads of spines in the tree recursively.
Having this performed we can read the solution from a sparsity vector of the root.
Let $(s_i)_{i=1}^\ell$ be a spine with a head $v$ and let $u^i$ indicate
the sparsity vector for the child of $s_i$ being a head (i.e., the child with the smaller subtree).
If $s_i$ has less than two children we treat $u^i$ as a vector~$(0)$.

For an interval $[a, b] \subseteq [1, \ell]$ let
$u^{a, b} = u^a \oplus^{\max} u^{a+1} \oplus^{\max} \dots \oplus^{\max} u^b$ and $y^{a, b}[k]$ be the maximum
weight of a subtree of size $k$ rooted at $s_{a}$ and not containing $s_{b+1}$.
Let $c = \floor{\frac{a+b}{2}}$.
The $\oplus^{\max}$ operator is associative so $u^{a, b} = u^{a, c} \oplus^{\max} u^{c+1, b}$.
To compute the second vector we consider two cases: whether the optimal subtree contains $s_{c+1}$ or not.

\begin{align}
    \label{eq:spine-binary} 
y^{a, b}[k] &= \max\bigg[y^{a, c}[k],\quad \sum_{i=a}^c x(s_i) 
    + \max_{k_1 + k_2 = k - (c - a + 1)} \Big(u^{a, c}[k_1] + y^{c+1, b}[k_2]
    \Big)\bigg] \\
&= \max\bigg[y^{a, c}[k],\quad \sum_{i=a}^c x(s_i) 
    + \Big(u^{a, c} \oplus^{\max} y^{c+1, b}\Big)\big[k - (c - a + 1)\big]
    \bigg]  \nonumber
\end{align}

Using the presented formulas we reduce the problem of computing $x^v = y^{1, \ell}$ to
subproblems for intervals $[1, \frac{\ell}{2}]$ and $[\frac{\ell}{2}+1, \ell]$ and results are merged
with two $(\max,+)$-convolutions.
Proceeding further we obtain $\log{\ell}$ levels of recursion.
Since there are $\Oh(\log n)$ spines on a path from a leaf to the root,
the whole computation tree has $\Oh(\log^2 n)$ layers, each node being
expressed as a pair of convolutions on vectors from its children.
Each vertex of the graph occurs in at most $\log n$ convolutions so the
sum of convolution sizes is $\Oh(n\log n)$.

In order to deal with the tail variant we consider a dual sparsity vector
$(\overline{x}^v[0], \overline{x}^v[1], \dots, \overline{x}^v[m])$,
where $\overline{x}^v[i]$ stands for the total weight of the subtree rooted at $v$ minus $x^v[i]$.
The dual sparsity vector of $v$ equals the \minconv of the vectors for the children of $v$.
We can use an analog of equation (\ref{eq:spine-binary}) and also express
the problem as a computation tree based on convolutions.

We take advantage of Theorem~\ref{approx-minconv} to perform each convolution with a relative error $\delta$.
The formula (\ref{eq:spine-binary}) contains an additive term $\sum_{i=a}^c x(s_i)$ but this can only decrease the relative error.
The cumulative relative error is bounded by $(1 - \delta)^{\log^2 n}$ for head approximation and $(1 + \delta)^{\log^2 n}$ for tail approximation, therefore setting $\delta = \Theta(\eps/\log^2{n})$ guarantees that the sparsity vector for the root is burdened with relative error at most $\eps$.

The sum of running times for all convolutions is
$\Oh\big(T(n, W, \delta)\log n\big)$,
what gives the postulated running time for the whole algorithm.
In order to retrieve the solution for a given $k$, we need to find the 
pair of indices that produced the value of the $k$-th index of the last convolution.
Then we proceed recursively and traverse back the computation tree.
Since finding $\argmax$ and $\argmin$ can be performed in linear time,
the total time of analyzing all convolutions is $\Oh(n\log{n})$.
\end{proof}

\section{$\Ot(n+1/\eps)$ approximation algorithm for \IIIsum}
\label{sec:ksum}

In the abstract we have claimed that our result for \partition constitutes the first
approximation algorithm for NP-hard problem that breaks the quadratic barrier. However
this is not necessary the case for the problems in P. In this section we will
show an $\Ot(n+1/\eps)$ approximation algorithm for \IIIsum and prove accompanying
lower bound under a reasonable assumption. To the best of our knowledge, this
is also the first nontrivial linear
approximation algorithm for a natural problem.

\defproblem{\ksum}
{\label{def:ksum}Sets $A_1,A_2,\ldots,A_{k-1}, S$, each with cardinality at most $n$.}
{Decide if there is a tuple $(a_1,\ldots,a_{k-1},s) \in A_1\times \ldots\times A_{k-1} \times S$ such that $a_1+\ldots +a_{k-1}=s$.}

The $\IIIsum$ problem is a special case of \ksum for $k=3$.  The \IIIsum is one
of the most notorious problems with a quadratic running time
and has been widely accepted as
a hardness assumption (see~\cite{ipec-survey} for overview).
The fastest known algorithm for \IIIsum is slightly subquadratic: \citet{3sum1} gave an
$\Oh(n^2(\log{\log{n}}/\log{n})^{2/3})$-time deterministic algorithm and then
independently~\citet{3sum2} and~\citet{3sum3} improved this result by presenting an
$\Oh(n^2\log{\log{n}}/\log{n})$-time algorithm.

The approximation variant for \IIIsum was considered by \citet{3sum-approx} who
showed a deterministic $\Ot(\frac{n}{\eps})$ algorithm as a byproduct of finding
longest approximate periodic patterns. If we are not interested in exact
solution, the~\citet{3sum-approx} algorithm is polynomially faster than the best
exact algorithm for \IIIsum. In this section we show how to solve
\IIIsum approximately in time $\Ot(n + 1/\eps)$ time and prove this tight up to
the polylogarithmic factors.

\defproblem{Approximate \IIIsum (\cite{3sum-approx})}
{\label{aksum} Three sets $A$, $B$, $C$ of positive integers, each with cardinality at most $n$.}
{The algorithm:
\begin{itemize}
    \item concludes that no triple $(a,b,c) \in A\times B\times C$ with $a+b=c$
        exists, or
    \item it outputs a triple $(a,b,c) \in A\times B \times C$ with $a+b \in [c/(1+\eps),\,c(1+\eps)].$
\end{itemize}}

This definition generalizes to \ksum, however we are unaware about any
previous works on approximate \ksum.

\subsection{Faster approximation algorithm for \IIIsum}

In this section we present an $\Ot(n+1/\eps)$-time approximation scheme
for \IIIsum problem. We use a technique from Section~\ref{approximate-minconv},
where we gave the fast approximation algorithm for \minconv. As previously, we
start with a fast $\Ot(n+W)$ exact algorithm
and then utilize rounding to get an approximation algorithm. In the
Section~\ref{lower-bounds} we will show a conditional optimality of this result.

\subsubsection{Exact $\Ot(n+W)$ algorithm for \IIIsum}

Let $W$ denote the upper bound on the integers in the sets $A$,$B$ and $C$.
The exact $\Ot(n+W)$-time algorithm for \IIIsum is already well known~\cite{cormen,chan}.
In here we will place the proof for completeness. For formal reasons
we need to take care of the special symbol $\infty$. What is more, we will
generalize this result to \ksum.

\begin{theorem}[Based on~\cite{cormen,chan}]
    \label{exact-ksum}
    The \ksum can be solved deterministically in $\Ot(kn+kW\log{W})$ time and
    $\Ot(kn+W)$ space.
\end{theorem}

\begin{proof}
    We will encode the numbers in the sets as binary arrays of size $\Oh(W)$ and
    iteratively perform fast convolution using FFT. Because we will use only
    $\Oh(1)$ tables at once, the space complexity will not depend on $k$. At the
    end we will need to check if any entry in the final array is in $S$.

    \paragraph{Encoding:} We iterate for every set $A_1,\ldots,A_{k-1}$ and
    for $l$-th iteration encode it as a binary vector $V$ of length $W+1$, such that:
    \begin{displaymath}
        V_l[i] = \begin{cases}
            1 &\text{iff}\; t \in A_l\\
            0 & \text{otherwise}
        \end{cases}
    \end{displaymath}
    to save space we will use only one $V_l$ vector at the time. The encoding
    can be done in $\Oh(n+W)$ time. If the special symbol $\infty\in A_l$
    appears then we simply discard it.

    \paragraph{FFT:} We want to perform a convolution with FFT on all vectors $V_l$. We do it
    one at a time and discard all elements larger than $W$.
    Let $U_l$ be the result of up to $l$-th iteration. We know that the
    proper polynomial is $U_l(x)$ $ = \sum_{(a_1,\ldots,a_l) \in
    (A_1\times\ldots \times A_l)}$ $ x^{a_1+\ldots a_l}$.  
    And if we multiply it by the polynomial $V_{l+1} = \sum_{a_{l+1} \in A_{l+1}}
    x^{a_{l+1}}$, we get $U_{l+1}(x) =$ $ \sum_{(a_1,\ldots,a_{l+1}) \in
    (A_1\times\ldots \times A_{l+1})} $ $x^{a_1+\ldots a_{l+1}}$.

    Hence at the end we obtain the vector $V_{k-1}$ that encodes all the sums of
    elements in subsets truncated up to $W$ place.

    \paragraph{Comparing} At the end we need to get the binary vector for $S$
    and compare it with the resulting vector $V_{k-1}$. 

    \paragraph{Time and Space} We did $k$ iterations. In each of them we
    transformed a set into a vector in time $\Oh(n)$. The fast convolution works in
    $\Oh(T\log{T})$ by using FFT.
    Hence, the running time is $\Oh(kn + kW\log{W})$. Algorithm needs $\Oh(nk)$
    space to encode input and $\Oh(W)$ space to store binary vectors.
\end{proof}

\subsubsection{Approximation algorithm}

Next we will use an exact algorithm to propose the fast approximation. We will
use the same reasoning as in Section~\ref{approx-minconv}.

\begin{lemma}
    \label{reduction-ksum}
    Assume the \ksum can be solved exactly in $T(n,k,W)$ time. Then approximate
    \ksum can be solved in $\Oh((T(n,k,k/\eps) + nk)\log{W})$ time.
\end{lemma}

Because the proof is just a small modification of the Lemma~\ref{approx-reduction} we have
included it in Section~\ref{proof-approximate-ksum}. At the end we need to
connect the exact algorithm from Lemma~\ref{exact-ksum} and the reduction from
Lemma~\ref{reduction-ksum}.

\begin{theorem}
    There is a deterministic algorithm for $(1+\eps)$-approximate \ksum running
    in $\Oh(nk\log{W} + \frac{k^2}{\eps}\log{\frac{k}{\eps}}\log{W})$ time.
\end{theorem}

\begin{proof}
    From Lemma~\ref{exact-ksum} the running time of \ksum is $T(n,k,W) = \Oh(nk
    + kW\log{W})$. Applying this running time to the reduction in
    Lemma~\ref{reduction-ksum} results in the claimed running time, because the
    $\Oh(nk)$ term is dominated by $\Oh(nk\log{W})$ term in the reduction. 
\end{proof}

To get an approximate algorithm for \IIIsum we set $k=3$.

\begin{corollary}
    \label{approx-3sum}
    The approximate \IIIsum can be solved deterministically in $\Oh((n +
    \frac{1}{\eps} \log{\frac{1}{\eps}} )\log{W}))$ time.
\end{corollary}

\subsubsection{Proof of Lemma~\ref{reduction-ksum}}
\label{proof-approximate-ksum}

\begin{algorithm}
    \caption{$\textsc{ApproximateKSum}(a_1,a_2,\ldots, a_{k-1}, s, \eps)$. We use a shorten notation to
    transform all elements in the sequences $a_l[i]$ and $s[i]$.}
	\label{alg:approx-ksum}
\begin{algorithmic}[1]
    \State $\text{Output}[i] = \infty$
    \For {$l=2\lceil\log{W}\rceil,\ldots,0$}
      \State $q := 2^l$ 
      \For {$l=1\ldots k-1$}
          \State $a_l'[i] = \ceil{ \frac{k a_l[i]}{q \eps} }$ 
          \If {$a_l'[i] > \ceil{4k/\eps}$} 
              \State $a_l'[i] = \infty$
          \EndIf
      \EndFor
      \State $s'[i] = \ceil{ \frac{k s[i]}{q \eps} }$ 
      \If {$s'[i] > \ceil{4k/\eps}$}
         \State $s'[i] = \infty$
      \EndIf
      \If {$\text{runExactKsum}(a_1',\ldots,a'_{k-1},s')$} 
        \State \Return True
       \EndIf
    \EndFor
    \State \Return False
\end{algorithmic}
\end{algorithm}

\begin{proof}
    The proof basically follows the approach approximating \minconv in
    Lemma~\ref{approx-reduction}. Assume, that there is some number $s$, for
    each there exists a tuple $(a_1, a_2, \ldots, a_{k-1}) \in A_1 \times \ldots
    A_{k-1}$, that $s < \sum_{i=1}^k a_i < s(1+\eps)$. Then look at
    Algorithm~\ref{alg:approx-ksum} in which we iterate precision parameter $q$. Hence
    there is some $q$, such that $q \le s < 2q$. From
    Lemma~\ref{knumbers-rounding} we know, that then, we can round the numbers
    $a'_i = \ceil{\frac{k a_i}{q\eps}}$ and then their sum should be
    approximately:
    \begin{displaymath}
        \sum_{i=1}^k a_i \le \sum \ceil{\frac{k a_i}{q\eps}} < (1+\eps)\sum_{i=1}^k a_i
    \end{displaymath}
    So if there is some number $s \in S$, then \textsc{ApproximateKSum}
    algorithm would find a tuple, that sum up to $s' \in [s, (1+\eps)s]$.

    From the other hand, if for all $s \in S$ no tuple sums up to $[s,
    (1+\eps)s]$ then our \textsc{ApproximateKSum} can also return YES. It is
    because before rounding items could sum up to something in $[(1-\eps)s,s]$
    (see Section~\ref{sec:techniques}).
    However, if there for such a parameter there always exists a precision parameter $q$, that $q\le s <2q$. Then 
    rounding the numbers according to Lemma~\ref{knumbers-rounding} gives only
    $(1 \pm \eps)$ error and they cannot sum to $\ceil{\frac{ks}{q\eps}}$. Hence if
    for all $s \in S$ no tuple sums up to $[(1-\eps)s, (1+\eps)s]$ then our
    \textsc{ApproximateKSum} will return NO.

    A technicality is hidden in the Definition~\ref{aksum} we need to return
    approximation of the form $\text{OPT}/(1+\eps) \le \text{OUR} \le
    \text{OPT}(1+\eps)$ but note that $\frac{1}{1+\eps} \approx 1-\eps$ so we
    can take care of if by adjusting $\eps$.
\end{proof}

\section{Conditional Lower Bounds}
\label{lower-bounds}

Proving conditional lower bounds in P under a plausible assumption  is
a very active line of
research~\cite{ipec-survey,hard-as-cnf-sat,batch-ov,hs1,subset-seth,itcs2017}.
One of the first problems
with a truly subquadratic running time ruled out was
\textsc{EditDistance}~\cite{edit-distance}.
It admits a linear time approximation
algorithm for $\eps = 1/\sqrt{n}$, that follows from the exact $\Oh(n+d^2)$
algorithm.
Subsequently, this (linear-time) approximation factor was improved by~\citet{edit-distance2}
to $n^{3/7+o(1)}$, then by~\citet{edit-distance3} to $n^{1/3+o(1)}$, and most
recently~\citet{edit-distance4} proposed an $\Oh(n^{1+\eps})$-time algorithm with factor
$(\log{n})^{\Oh(1/\eps)}$ for every fixed $\eps>0$. From the other hand~\citet{itcs2017}
ruled out a truly subquadratic PTAS for \textsc{Longest Common Subsequence}
using {circuit lower bounds}. Our results are somehow of similar
flavor to this line of research.

\subsection{Conditional Lower Bound for Approximate \IIIsum}

We have shown an approximate algorithm for \IIIsum running in $\Ot(n+1/\eps)$ time.
Is this the best we can hope for? Perhaps one could imagine an
$\Ot(n+1/\sqrt{\eps})$ time algorithm. In this subsection we 
rule out such a possibility and prove the optimality of
Theorem~\ref{approx-3sum}.

To show the conditional lower bound we will assume the hardness of the exact
\IIIsum. The \IIIsum conjecture says, that the $\Ot(n^2)$ algorithm is
essentially the best we can hope for up to subpolynomial factors.

\begin{conjecture}[\IIIsum conjecture~\cite{ipec-survey}]
    \label{con:3sum}
    In the Word RAM model with $\Oh(\log{n})$ bit words, any algorithm requires
    $\Omega(n^{2-o(1)})$ time in expectation to determine whether given set $S \subset
    \{-n^{3+o(1)},\ldots,n^{3+o(1)}\}$ of size $n$ contains three distinct elements $a,b,c$ such that $a+b=c$.
\end{conjecture}

This definition of \IIIsum in~\cite{ipec-survey} is equivalent to
the one in Section~\ref{problems-definitions} (see discussion in~\cite{patrascu}). What is more, solving
\IIIsum with only polynomially bounded numbers can be reduced to solving it with the upper bound $W = \Oh(n^3)$~\cite{patrascu}.
\IIIsum can be solved in subquadratic time when $W = \Oh(n^{2-\delta})$ via FFT,
but doing so assuming only $W = \Oh(n^{2})$ constitutes a major open
problem. \citet{strong-3sum} have considered it as
a yet another hardness assumption.

\begin{conjecture}[Strong-3SUM conjecture~\cite{strong-3sum}]
    \label{con:strong-3sum}
    \IIIsum on a set of $n$ integers in the domain of $\{-n^2,\ldots,n^2\}$
    requires time $\Omega(n^{2-o(1)})$.
\end{conjecture}

\begin{theorem}
    Assuming the Strong-3SUM conjecture, there is no $\Ot(n+1/\eps^{1-\delta})$
    algorithm for $(1+\eps)$-approximate \IIIsum, for any constant $\delta > 0$.
\end{theorem}

\begin{proof}
Consider the exact variant of \IIIsum within the domain $\{-n^2,\ldots,n^2\}$.
We can assume that the numbers are divided into sets $A, B, C$
and we can restrict ourselves to triples $a\in A,\, b\in B$, $c\in C$~\cite{patrascu}.
We add $n^2 + 1$ to all numbers in $A \cup B$ and likewise $2n^2+2$ to numbers in $C$ to obtain an equivalent instance with all input numbers greater than 0 and $W=\Oh(n^2)$.

    Suppose, that for some small $\delta > 0$ the approximate \IIIsum admits an
    $\Ot(n+1/\eps^{1-\delta})$-algorithm.
    We can use it to solve the problem above exactly by setting $\eps = \frac{1}{2W} = \Omega(n^{\frac{1}{2}})$. 
    The running time of the exact algorithm is strongly subquadratic, namely
    $\Ot(n+1/\eps^{1-\delta}) = \Ot(n^{2-2\delta})$.
This contradicts the Strong-3SUM conjecture. 
\end{proof}

\subsection{Conditional Lower Bounds for \knapsack-type Problems}

The conditional lower bounds for \knapsack and \uknapsack are
corollaries from~\cite{icalp2017}. We commence by introducing the
main theorem from that work, truncated to problems that are of interest to us.

\begin{theorem}[Theorem 2 from~\cite{icalp2017}]
    The following statements are equivalent:
    \begin{enumerate}
        \item There exists an $\Oh(n^{2-\varepsilon})$ algorithm for \minconv for some $\varepsilon>0$.
        \item There exists an $\Oh\left((n+t)^{2-\varepsilon}\right)$ algorithm for \uknapsack for some $\varepsilon>0$.
        \item There exists an $\Oh\left((n+t)^{2-\varepsilon}\right)$ algorithm for \knapsack for some $\varepsilon>0$.
    \end{enumerate}
    We allow randomized algorithms.
\end{theorem}

\begin{conjecture}[\minconv conjecture~\cite{icalp2017}]
    \label{minconv-conjecture}
    Any algorithm computing \minconv requires $\Omega(n^{2-o(1)})$ running time.
\end{conjecture}

Basically \cite[Theorem 2]{icalp2017} says that assuming the \minconv
conjecture both \uknapsack and \knapsack require $\Omega((n+t)^{2-o(1)})$ time.
The pseudo-polynomial algorithm for \knapsack running in time $\Oh(nt)$ can be modified to work
in time $\Oh(nv)$, where $v$ is an upper bound on value of the
solution. In similar spirit, the reductions from \cite{icalp2017}
can use a hypothetical $\Oh\left( \left( n+v \right)^{2-\delta} \right)$ algorithm for
\knapsack or \uknapsack to get a subquadratic algorithm for \minconv (modify Theorem 4
from \cite{icalp2017}).

\begin{observation}[\cite{icalp2017}]

    For any constant $\delta > 0$, an exact algorithm for \knapsack or
    \uknapsack with $\Oh\left( \left( n+v \right)^{2-\delta} \right)$ running time
    would refute the \minconv conjecture.

\end{observation}

We need this modification because in the definition of
FPTAS for \knapsack we consider relative error with respect to the optimal value (not weight).
We can use a hypothetical faster approximation algorithm to get a faster
pseudo-polynomial exact algorithm, what would contradict
the \minconv conjecture. More formally:

\begin{theorem}[restated Theorem~\ref{knapsack-lower-bound}]
    For any constant $\delta > 0$, obtaining a \weak for \knapsack or \uknapsack with
    $\Oh((n+1/\eps)^{2-\delta})$ running time would refute the \minconv conjecture.
\end{theorem}

\begin{proof}
    Suppose, that for some $\delta > 0$ we have a \weak for \knapsack (or \uknapsack)
    with running time $\Oh\left((n+1/\eps)^{2-\delta} \right)$.
    If we set $\eps =
    2/v$, then the approximation algorithm would solve the exact problem because the absolute error gets bounded by $1/2$.
    By \cite[Theorem 2]{icalp2017}
    we know that such an algorithm contradicts the \minconv conjecture. The claim follows.
\end{proof}

A similar argument works for the \ssum problem.
\citet{subset-seth} showed that assuming SETH there can be no $\Oh(t^{1-\delta}
\mathrm{poly}(n))$ algorithm for \ssum (\citet{hard-as-cnf-sat} obtained the
same lower bound before but assuming the SetCover conjecture).

\begin{theorem}[Conditional Lower Bound for approximate \ssum]
    For any constant $\delta > 0$, a \weak for \ssum with running time $\Oh\left(\mathrm{poly}(n) \left( \frac{1}{\eps}
    \right)^{1-\delta}\right)$ would refute SETH and SetCover conjecture.
\end{theorem}

\begin{proof}
    We set $\eps = 2/t$ and obtain an algorithm
    solving the exact \ssum, because all numbers are integers and the absolute error is at most $1/2$.
    The running time is $\Oh(t^{1-\delta} \mathrm{poly}(n))$,
    what refutes SETH due
    to~\cite{subset-seth} and the SetCover conjecture due to~\cite{hard-as-cnf-sat}.
\end{proof}
\section{Conclusion and Open Problems}
\label{conclusion}

In this paper we study the complexity of the \knapsack, \ssum and \partition. In
the exact setting, if we are only concerned about the dependence on $n$, \knapsack
and \ssum were already known to be equivalent up to the polynomial factors.
\citet[Theorem 2]{nederlof12} showed, that if there exists an
exact algorithm for \ssum working in $\Os(T(n))$ time and $\Os(S(n))$ space,
then we can construct an algorithm for \knapsack working in the same $\Os(T(n))$ time and
$\Os(S(n))$ space. 
In contrast, in the realm of pseudo-polynomial time complexity, \ssum seems to 
be simpler than \knapsack (see~\citet{bringmann-soda,icalp2017}). In this paper, we
show similar separation for \knapsack and \partition in the approximation setting.

After
this paper was announced, Bringmann~\cite{bringmann-com} showed
that the current $\Ot(n+1/\eps^2)$ algorithm for \ssum is optimal assuming
\minconv conjecture. Can we improve the approximation algorithm for \knapsack to an $\Ot(n+1/\eps^2)$ 
and match the quadratic lower bound?

It also remains open whether \IIIsum and \minconv admit FPTAS algorithms with no dependence on $W$. To add weight to 
this open problem, note that it is this issue that makes the FPTAS algorithms for \sparsity
inefficient in practice.

Closing the time complexity gap for \partition is another open problem, either by improving the
$\Ot((n+1/\eps)^{5/3}))$ FPTAS or the $\Omega((n+1/\eps)^{1-o(1)})$
conditional lower bound. It is worth noting, that if the Freiman's Conjecture~\cite{galil} is true, then our
techniques would automatically lead to even faster FPTAS for \partition.

Finally, one can also ask whether randomization is necessary to obtain subquadratic FPTAS for \partition. 
We believe that the randomized building blocks can be replaced with deterministic algorithms by
\citet{kellerer-subsetsum} and \citet{koiliaris-soda}.

\section{Acknowledgements}

This work is part of the project TOTAL that has received funding from the European
Research Council (ERC) under the European Union’s Horizon 2020 research and
innovation programme (grant agreement No 677651). Karol W\k{e}grzycki is supported by the grants 2016/21/N/ST6/01468 and
2018/28/T/ST6/00084 of the Polish National Science Center. We would
like to thank Marek Cygan, Artur Czumaj, Zvi Galil, Oded Margalit and Piotr Sankowski for
helpful discussions. Also we are grateful to the organizers and participants of
the \textit{Bridging Continuous and Discrete Optimization} program at the Simons Institute for the Theory
of Computing, especially Aleksander M\k{a}dry.

\bibliographystyle{plainnat}
\bibliography{subsetsum}

\begin{thebibliography}{60}
\providecommand{\natexlab}[1]{#1}
\providecommand{\url}[1]{\texttt{#1}}
\expandafter\ifx\csname urlstyle\endcsname\relax
  \providecommand{\doi}[1]{doi: #1}\else
  \providecommand{\doi}{doi: \begingroup \urlstyle{rm}\Url}\fi

\bibitem[Abboud and Backurs(2017)]{itcs2017}
Amir Abboud and Arturs Backurs.
\newblock Towards hardness of approximation for polynomial time problems.
\newblock In Christos~H. Papadimitriou, editor, \emph{8th Innovations in
  Theoretical Computer Science Conference, {ITCS} 2017, January 9-11, 2017,
  Berkeley, CA, {USA}}, volume~67 of \emph{LIPIcs}, pages 11:1--11:26. Schloss
  Dagstuhl - Leibniz-Zentrum fuer Informatik, 2017.
\newblock ISBN 978-3-95977-029-3.
\newblock URL \url{http://www.dagstuhl.de/dagpub/978-3-95977-029-3}.

\bibitem[Abboud et~al.(2015)Abboud, Williams, and Yu]{batch-ov}
Amir Abboud, Ryan Williams, and Huacheng Yu.
\newblock More applications of the polynomial method to algorithm design.
\newblock In \emph{Proceedings of the Twenty-sixth Annual ACM-SIAM Symposium on
  Discrete Algorithms}, SODA '15, pages 218--230, Philadelphia, PA, USA, 2015.
  Society for Industrial and Applied Mathematics.

\bibitem[Abboud et~al.(2016)Abboud, Williams, and Wang]{hs1}
Amir Abboud, Virginia~Vassilevska Williams, and Joshua~R. Wang.
\newblock Approximation and fixed parameter subquadratic algorithms for radius
  and diameter in sparse graphs.
\newblock In Robert Krauthgamer, editor, \emph{Proceedings of the
  Twenty-Seventh Annual {ACM-SIAM} Symposium on Discrete Algorithms, {SODA}
  2016, Arlington, VA, USA, January 10-12, 2016}, pages 377--391. {SIAM}, 2016.

\bibitem[Abboud et~al.(2019)Abboud, Bringmann, Hermelin, and
  Shabtay]{subset-seth}
Amir Abboud, Karl Bringmann, Danny Hermelin, and Dvir Shabtay.
\newblock Seth-based lower bounds for subset sum and bicriteria path.
\newblock \emph{arXiv preprint arXiv:1704.04546, to appear at SODA 2019}, 2019.

\bibitem[Andoni et~al.(2010)Andoni, Krauthgamer, and Onak]{edit-distance4}
Alexandr Andoni, Robert Krauthgamer, and Krzysztof Onak.
\newblock Polylogarithmic approximation for edit distance and the asymmetric
  query complexity.
\newblock In \emph{51th Annual {IEEE} Symposium on Foundations of Computer
  Science, {FOCS} 2010, October 23-26, 2010, Las Vegas, Nevada, {USA}}, pages
  377--386. {IEEE} Computer Society, 2010.

\bibitem[Austrin et~al.(2013)Austrin, Kaski, Koivisto, and
  M{\"{a}}{\"{a}}tt{\"{a}}]{random-ss1}
Per Austrin, Petteri Kaski, Mikko Koivisto, and Jussi M{\"{a}}{\"{a}}tt{\"{a}}.
\newblock Space-time tradeoffs for subset sum: An improved worst case
  algorithm.
\newblock In Fedor~V. Fomin, Rusins Freivalds, Marta~Z. Kwiatkowska, and David
  Peleg, editors, \emph{Automata, Languages, and Programming - 40th
  International Colloquium, {ICALP} 2013, Riga, Latvia, July 8-12, 2013,
  Proceedings, Part {I}}, volume 7965 of \emph{Lecture Notes in Computer
  Science}, pages 45--56. Springer, 2013.

\bibitem[Austrin et~al.(2015)Austrin, Kaski, Koivisto, and
  Nederlof]{random-ss2}
Per Austrin, Petteri Kaski, Mikko Koivisto, and Jesper Nederlof.
\newblock Subset sum in the absence of concentration.
\newblock In Ernst~W. Mayr and Nicolas Ollinger, editors, \emph{32nd
  International Symposium on Theoretical Aspects of Computer Science, {STACS}
  2015, March 4-7, 2015, Garching, Germany}, volume~30 of \emph{LIPIcs}, pages
  48--61. Schloss Dagstuhl - Leibniz-Zentrum fuer Informatik, 2015.

\bibitem[Austrin et~al.(2016)Austrin, Kaski, Koivisto, and
  Nederlof]{random-ss3}
Per Austrin, Petteri Kaski, Mikko Koivisto, and Jesper Nederlof.
\newblock Dense subset sum may be the hardest.
\newblock In Nicolas Ollinger and Heribert Vollmer, editors, \emph{33rd
  Symposium on Theoretical Aspects of Computer Science, {STACS} 2016, February
  17-20, 2016, Orl{\'{e}}ans, France}, volume~47 of \emph{LIPIcs}, pages
  13:1--13:14. Schloss Dagstuhl - Leibniz-Zentrum fuer Informatik, 2016.

\bibitem[Backurs and Indyk(2015)]{edit-distance}
Arturs Backurs and Piotr Indyk.
\newblock Edit distance cannot be computed in strongly subquadratic time
  (unless {SETH} is false).
\newblock In Rocco~A. Servedio and Ronitt Rubinfeld, editors, \emph{Proceedings
  of the Forty-Seventh Annual {ACM} on Symposium on Theory of Computing, {STOC}
  2015, Portland, OR, USA, June 14-17, 2015}, pages 51--58. {ACM}, 2015.

\bibitem[Backurs et~al.(2017)Backurs, Indyk, and Schmidt]{tree-sparsity}
Arturs Backurs, Piotr Indyk, and Ludwig Schmidt.
\newblock Better approximations for tree sparsity in nearly-linear time.
\newblock In Philip~N. Klein, editor, \emph{Proceedings of the Twenty-Eighth
  Annual {ACM-SIAM} Symposium on Discrete Algorithms, {SODA} 2017, Barcelona,
  Spain, Hotel Porta Fira, January 16-19}, pages 2215--2229. {SIAM}, 2017.

\bibitem[Bansal et~al.(2017)Bansal, Garg, Nederlof, and Vyas]{bansal17}
Nikhil Bansal, Shashwat Garg, Jesper Nederlof, and Nikhil Vyas.
\newblock Faster space-efficient algorithms for subset sum and k-sum.
\newblock In Hamed Hatami, Pierre McKenzie, and Valerie King, editors,
  \emph{Proceedings of the 49th Annual {ACM} {SIGACT} Symposium on Theory of
  Computing, {STOC} 2017, Montreal, QC, Canada, June 19-23, 2017}, pages
  198--209. {ACM}, 2017.

\bibitem[Bar{-}Yossef et~al.(2004)Bar{-}Yossef, Jayram, Krauthgamer, and
  Kumar]{edit-distance2}
Ziv Bar{-}Yossef, T.~S. Jayram, Robert Krauthgamer, and Ravi Kumar.
\newblock Approximating edit distance efficiently.
\newblock In \emph{45th Symposium on Foundations of Computer Science {(FOCS}
  2004), 17-19 October 2004, Rome, Italy, Proceedings}, pages 550--559. {IEEE}
  Computer Society, 2004.

\bibitem[Baran et~al.(2008)Baran, Demaine, and Patrascu]{patrascu}
Ilya Baran, Erik~D. Demaine, and Mihai Patrascu.
\newblock Subquadratic algorithms for 3sum.
\newblock \emph{Algorithmica}, 50\penalty0 (4):\penalty0 584--596, 2008.

\bibitem[Batu et~al.(2006)Batu, Erg{\"{u}}n, and Sahinalp]{edit-distance3}
Tugkan Batu, Funda Erg{\"{u}}n, and S{\"{u}}leyman~Cenk Sahinalp.
\newblock Oblivious string embeddings and edit distance approximations.
\newblock In \emph{Proceedings of the Seventeenth Annual {ACM-SIAM} Symposium
  on Discrete Algorithms, {SODA} 2006, Miami, Florida, USA, January 22-26,
  2006}, pages 792--801. {ACM} Press, 2006.

\bibitem[Bellman(1957)]{bellman}
Richard Bellman.
\newblock \emph{Dynamic Programming}.
\newblock Princeton University Press, Princeton, NJ, USA, 1957.

\bibitem[Bhalgat et~al.(2011)Bhalgat, Goel, and Khanna]{stochastic-knapsack}
Anand Bhalgat, Ashish Goel, and Sanjeev Khanna.
\newblock Improved approximation results for stochastic knapsack problems.
\newblock In Dana Randall, editor, \emph{Proceedings of the Twenty-Second
  Annual {ACM-SIAM} Symposium on Discrete Algorithms, {SODA} 2011, San
  Francisco, California, USA, January 23-25, 2011}, pages 1647--1665. {SIAM},
  2011.

\bibitem[Bringmann(2017)]{bringmann-soda}
Karl Bringmann.
\newblock A near-linear pseudopolynomial time algorithm for subset sum.
\newblock In \emph{Proceedings of the Twenty-Eighth Annual ACM-SIAM Symposium
  on Discrete Algorithms}, SODA '17, pages 1073--1084, Philadelphia, PA, USA,
  2017. Society for Industrial and Applied Mathematics.

\bibitem[Bringmann(April 2018)]{bringmann-com}
Karl Bringmann.
\newblock personal communication, April 2018.

\bibitem[Chan(2018)]{chan-knapsack}
Timothy~M. Chan.
\newblock Approximation schemes for 0-1 knapsack.
\newblock In Raimund Seidel, editor, \emph{1st Symposium on Simplicity in
  Algorithms, {SOSA} 2018, January 7-10, 2018, New Orleans, LA, {USA}},
  volume~61 of \emph{{OASICS}}, pages 5:1--5:12. Schloss Dagstuhl -
  Leibniz-Zentrum fuer Informatik, 2018.

\bibitem[Chan and Lewenstein(2015)]{chan}
Timothy~M. Chan and Moshe Lewenstein.
\newblock Clustered integer 3sum via additive combinatorics.
\newblock In \emph{Proceedings of the Forty-seventh Annual ACM Symposium on
  Theory of Computing}, STOC '15, pages 31--40, New York, NY, USA, 2015. ACM.

\bibitem[Chatzigiannakis et~al.(2017)Chatzigiannakis, Indyk, Kuhn, and
  Muscholl]{DBLP:conf/icalp/2017}
Ioannis Chatzigiannakis, Piotr Indyk, Fabian Kuhn, and Anca Muscholl, editors.
\newblock \emph{44th International Colloquium on Automata, Languages, and
  Programming, {ICALP} 2017, July 10-14, 2017, Warsaw, Poland}, volume~80 of
  \emph{LIPIcs}, 2017. Schloss Dagstuhl - Leibniz-Zentrum fuer Informatik.

\bibitem[Cormen(2009)]{cormen}
Thomas~H Cormen.
\newblock \emph{Introduction to algorithms}.
\newblock MIT press, 2009.

\bibitem[Cygan et~al.(2012)Cygan, Dell, Lokshtanov, Marx, Nederlof, Okamoto,
  Paturi, Saurabh, and Wahlstrom]{hard-as-cnf-sat}
Marek Cygan, Holger Dell, Daniel Lokshtanov, Da\'niel Marx, Jesper Nederlof,
  Yoshio Okamoto, Ramamohan Paturi, Saket Saurabh, and Magnus Wahlstrom.
\newblock On problems as hard as cnf-sat.
\newblock In \emph{Proceedings of the 2012 IEEE Conference on Computational
  Complexity (CCC)}, CCC '12, pages 74--84, Washington, DC, USA, 2012. IEEE
  Computer Society.

\bibitem[Cygan et~al.(2017)Cygan, Mucha, Wegrzycki, and Wlodarczyk]{icalp2017}
Marek Cygan, Marcin Mucha, Karol Wegrzycki, and Michal Wlodarczyk.
\newblock On problems equivalent to (min, +)-convolution.
\newblock In  \citet{DBLP:conf/icalp/2017}, pages 22:1--22:15.

\bibitem[Freund(2017)]{3sum2}
Ari Freund.
\newblock Improved subquadratic 3sum.
\newblock \emph{Algorithmica}, 77\penalty0 (2):\penalty0 440--458, 2017.

\bibitem[Galil and Margalit(1991{\natexlab{a}})]{galil}
Zvi Galil and Oded Margalit.
\newblock An almost linear-time algorithm for the dense subset-sum problem.
\newblock \emph{{SIAM} J. Comput.}, 20\penalty0 (6):\penalty0 1157--1189,
  1991{\natexlab{a}}.

\bibitem[Galil and Margalit(1991{\natexlab{b}})]{galil-icalp}
Zvi Galil and Oded Margalit.
\newblock An almost linear-time algorithm for the dense subset-sum problem.
\newblock In Javier~Leach Albert, Burkhard Monien, and Mario
  Rodr{\'{\i}}guez{-}Artalejo, editors, \emph{Automata, Languages and
  Programming, 18th International Colloquium, ICALP91, Madrid, Spain, July
  8-12, 1991, Proceedings}, volume 510 of \emph{Lecture Notes in Computer
  Science}, pages 719--727. Springer, 1991{\natexlab{b}}.
\newblock ISBN 3-540-54233-7.

\bibitem[Galil and Margalit(2017)]{galil-comunication}
Zvi Galil and Oded Margalit.
\newblock personal communication, 2017.

\bibitem[Gens and Levner(1979)]{levner79}
George Gens and Eugene Levner.
\newblock Computational complexity of approximation algorithms for
  combinatorial problems.
\newblock In Jir{\'{\i}} Becv{\'{a}}r, editor, \emph{Mathematical Foundations
  of Computer Science 1979, Proceedings, 8th Symposium, Olomouc,
  Czechoslovakia, September 3-7, 1979}, volume~74 of \emph{Lecture Notes in
  Computer Science}, pages 292--300. Springer, 1979.

\bibitem[Gens and Levner(1994)]{levner94}
George Gens and Eugene Levner.
\newblock A fast approximation algorithm for the subset-sum problem.
\newblock \emph{INFOR: Information Systems and Operational Research},
  32\penalty0 (3):\penalty0 143--148, 1994.

\bibitem[Gens and Levner(1980)]{levner80}
Georgii~V Gens and Eugenii~V Levner.
\newblock Fast approximation algorithms for knapsack type problems.
\newblock In \emph{Optimization Techniques}, pages 185--194. Springer, 1980.

\bibitem[Gens and Levner(1978)]{levner78}
GV~Gens and EV~Levner.
\newblock Approximation algorithm for some scheduling problems.
\newblock \emph{Engrg. Cybernetics}, 6:\penalty0 38--46, 1978.

\bibitem[Gfeller(2011)]{3sum-approx}
Beat Gfeller.
\newblock Finding longest approximate periodic patterns.
\newblock In Frank Dehne, John Iacono, and J{\"{o}}rg{-}R{\"{u}}diger Sack,
  editors, \emph{Algorithms and Data Structures - 12th International Symposium,
  {WADS} 2011, New York, NY, USA, August 15-17, 2011. Proceedings}, volume 6844
  of \emph{Lecture Notes in Computer Science}, pages 463--474. Springer, 2011.

\bibitem[Gold and Sharir(2017)]{3sum3}
Omer Gold and Micha Sharir.
\newblock Improved bounds for 3sum, k-sum, and linear degeneracy.
\newblock In Kirk Pruhs and Christian Sohler, editors, \emph{25th Annual
  European Symposium on Algorithms, {ESA} 2017, September 4-6, 2017, Vienna,
  Austria}, volume~87 of \emph{LIPIcs}, pages 42:1--42:13. Schloss Dagstuhl -
  Leibniz-Zentrum fuer Informatik, 2017.

\bibitem[Hayes(2002)]{partition3}
Brian Hayes.
\newblock Computing science: The easiest hard problem.
\newblock \emph{American Scientist}, 90\penalty0 (2):\penalty0 113--117, 2002.

\bibitem[Horowitz and Sahni(1974)]{meetinthemiddle}
Ellis Horowitz and Sartaj Sahni.
\newblock Computing partitions with applications to the knapsack problem.
\newblock \emph{J. {ACM}}, 21\penalty0 (2):\penalty0 277--292, 1974.

\bibitem[Howgrave{-}Graham and Joux(2010)]{random-ss4}
Nick Howgrave{-}Graham and Antoine Joux.
\newblock New generic algorithms for hard knapsacks.
\newblock In Henri Gilbert, editor, \emph{Advances in Cryptology - {EUROCRYPT}
  2010, 29th Annual International Conference on the Theory and Applications of
  Cryptographic Techniques, French Riviera, May 30 - June 3, 2010.
  Proceedings}, volume 6110 of \emph{Lecture Notes in Computer Science}, pages
  235--256. Springer, 2010.

\bibitem[Hsu and Umans(2017)]{strong-3sum}
Chloe Ching-Yun Hsu and Chris Umans.
\newblock On multidimensional and monotone k-sum.
\newblock \emph{To appear at MFCS 2017}, 2017.

\bibitem[Ibarra and Kim(1975)]{ibarra-kim}
Oscar~H. Ibarra and Chul~E. Kim.
\newblock Fast approximation algorithms for the knapsack and sum of subset
  problems.
\newblock \emph{J. {ACM}}, 22\penalty0 (4):\penalty0 463--468, 1975.

\bibitem[Jansen and Kraft(2015)]{fptas-uknapsack}
Klaus Jansen and Stefan Erich~Julius Kraft.
\newblock A faster {FPTAS} for the unbounded knapsack problem.
\newblock In Zsuzsanna Lipt{\'{a}}k and William~F. Smyth, editors,
  \emph{Combinatorial Algorithms - 26th International Workshop, {IWOCA} 2015,
  Verona, Italy, October 5-7, 2015, Revised Selected Papers}, volume 9538 of
  \emph{Lecture Notes in Computer Science}, pages 274--286. Springer, 2015.

\bibitem[J{\o}rgensen and Pettie(2014)]{3sum1}
Allan~Gr{\o}nlund J{\o}rgensen and Seth Pettie.
\newblock Threesomes, degenerates, and love triangles.
\newblock In \emph{55th {IEEE} Annual Symposium on Foundations of Computer
  Science, {FOCS} 2014, Philadelphia, PA, USA, October 18-21, 2014}, pages
  621--630. {IEEE} Computer Society, 2014.

\bibitem[Jr. and Lueker(1991)]{partition1}
Edward G.~Coffman Jr. and George~S. Lueker.
\newblock \emph{Probabilistic analysis of packing and partitioning algorithms}.
\newblock Wiley-Interscience series in discrete mathematics and optimization.
  Wiley, 1991.

\bibitem[Karp(1972)]{karp21}
Richard~M. Karp.
\newblock Reducibility among combinatorial problems.
\newblock In Raymond~E. Miller and James~W. Thatcher, editors,
  \emph{Proceedings of a symposium on the Complexity of Computer Computations,
  held March 20-22, 1972, at the {IBM} Thomas J. Watson Research Center,
  Yorktown Heights, New York.}, The {IBM} Research Symposia Series, pages
  85--103. Plenum Press, New York, 1972.

\bibitem[Karp(1975)]{karp75}
R.M. Karp.
\newblock The fast approximate solution to hard combinatorial problems.
\newblock \emph{Proceedings of the 6th Southeastern Conference on
  Combinatorics, Graph Theory and Computing}, pages 15--31, 1975.

\bibitem[Kellerer and Pferschy(2004)]{fptas-knapsack}
Hans Kellerer and Ulrich Pferschy.
\newblock Improved dynamic programming in connection with an {FPTAS} for the
  knapsack problem.
\newblock \emph{J. Comb. Optim.}, 8\penalty0 (1):\penalty0 5--11, 2004.

\bibitem[Kellerer et~al.(1997)Kellerer, Pferschy, and
  Speranza]{kellerer-subsetsum}
Hans Kellerer, Ulrich Pferschy, and Maria~Grazia Speranza.
\newblock An efficient approximation scheme for the subset-sum problem.
\newblock In Hon~Wai Leong, Hiroshi Imai, and Sanjay Jain, editors,
  \emph{Algorithms and Computation, 8th International Symposium, {ISAAC} '97,
  Singapore, December 17-19, 1997, Proceedings}, volume 1350 of \emph{Lecture
  Notes in Computer Science}, pages 394--403. Springer, 1997.

\bibitem[Kellerer et~al.(2004)Kellerer, Pferschy, and Pisinger]{knapsack-book}
Hans Kellerer, Ulrich Pferschy, and David Pisinger.
\newblock \emph{Knapsack problems}.
\newblock Springer, 2004.

\bibitem[Koiliaris and Xu(2017)]{koiliaris-soda}
Konstantinos Koiliaris and Chao Xu.
\newblock A faster pseudopolynomial time algorithm for subset sum.
\newblock In \emph{Proceedings of the Twenty-Eighth Annual ACM-SIAM Symposium
  on Discrete Algorithms}, SODA '17, pages 1062--1072, Philadelphia, PA, USA,
  2017. Society for Industrial and Applied Mathematics.

\bibitem[K{\"{u}}nnemann et~al.(2017)K{\"{u}}nnemann, Paturi, and
  Schneider]{kunnemann-icalp2017}
Marvin K{\"{u}}nnemann, Ramamohan Paturi, and Stefan Schneider.
\newblock On the fine-grained complexity of one-dimensional dynamic
  programming.
\newblock In  \citet{DBLP:conf/icalp/2017}, pages 21:1--21:15.

\bibitem[Lawler(1979)]{lawler79}
Eugene~L Lawler.
\newblock Fast approximation algorithms for knapsack problems.
\newblock \emph{Mathematics of Operations Research}, 4\penalty0 (4):\penalty0
  339--356, 1979.

\bibitem[Lipsky and Porat(2011)]{approx-pattern-matching}
Ohad Lipsky and Ely Porat.
\newblock Approximate pattern matching with the \emph{L}\({}_{\mbox{1}}\),
  \emph{L}\({}_{\mbox{2}}\) and \emph{L}\({}_{\mbox{{\(\infty\)}}}\) metrics.
\newblock \emph{Algorithmica}, 60\penalty0 (2):\penalty0 335--348, 2011.

\bibitem[Mathews(1896)]{first-knapsack}
George~B Mathews.
\newblock On the partition of numbers.
\newblock \emph{Proceedings of the London Mathematical Society}, 1\penalty0
  (1):\penalty0 486--490, 1896.

\bibitem[Merkle and Hellman(1978)]{partition2}
Ralph~C. Merkle and Martin~E. Hellman.
\newblock Hiding information and signatures in trapdoor knapsacks.
\newblock \emph{{IEEE} Trans. Information Theory}, 24\penalty0 (5):\penalty0
  525--530, 1978.

\bibitem[Mertens(2006)]{partition4}
Stephan Mertens.
\newblock The easiest hard problem: Number partitioning.
\newblock \emph{Computational Complexity and Statistical Physics}, 125\penalty0
  (2):\penalty0 125--139, 2006.

\bibitem[Nederlof et~al.(2012)Nederlof, van Leeuwen, and van~der
  Zwaan]{nederlof12}
Jesper Nederlof, Erik~Jan van Leeuwen, and Ruben van~der Zwaan.
\newblock Reducing a target interval to a few exact queries.
\newblock In Branislav Rovan, Vladimiro Sassone, and Peter Widmayer, editors,
  \emph{Mathematical Foundations of Computer Science 2012 - 37th International
  Symposium, {MFCS} 2012, Bratislava, Slovakia, August 27-31, 2012.
  Proceedings}, volume 7464 of \emph{Lecture Notes in Computer Science}, pages
  718--727. Springer, 2012.

\bibitem[Schroeppel and Shamir(1981)]{schroeppel}
Richard Schroeppel and Adi Shamir.
\newblock A t=o(2\({}^{\mbox{n/2}}\)), s=o(2\({}^{\mbox{n/4}}\)) algorithm for
  certain np-complete problems.
\newblock \emph{{SIAM} J. Comput.}, 10\penalty0 (3):\penalty0 456--464, 1981.

\bibitem[Sleator and Tarjan(1983)]{heavy-light}
Daniel~D. Sleator and Robert~Endre Tarjan.
\newblock A data structure for dynamic trees.
\newblock \emph{J. Comput. Syst. Sci.}, 26\penalty0 (3):\penalty0 362--391,
  June 1983.
\newblock ISSN 0022-0000.

\bibitem[Williams(2015)]{ipec-survey}
Virginia~Vassilevska Williams.
\newblock Hardness of easy problems: Basing hardness on popular conjectures
  such as the strong exponential time hypothesis (invited talk).
\newblock In Thore Husfeldt and Iyad~A. Kanj, editors, \emph{10th International
  Symposium on Parameterized and Exact Computation, {IPEC} 2015, September
  16-18, 2015, Patras, Greece}, volume~43 of \emph{LIPIcs}, pages 17--29.
  Schloss Dagstuhl - Leibniz-Zentrum fuer Informatik, 2015.

\bibitem[Woeginger(2000)]{woeginger}
Gerhard~J. Woeginger.
\newblock When does a dynamic programming formulation guarantee the existence
  of a fully polynomial time approximation scheme (fptas)?
\newblock \emph{{INFORMS} Journal on Computing}, 12\penalty0 (1):\penalty0
  57--74, 2000.

\bibitem[Zwick(1998)]{approx-apsp}
Uri Zwick.
\newblock All pairs shortest paths in weighted directed graphs -- exact and
  almost exact algorithms.
\newblock In \emph{39th Annual Symposium on Foundations of Computer Science,
  {FOCS} '98, November 8-11, 1998, Palo Alto, California, {USA}}, pages
  310--319. {IEEE} Computer Society, 1998.

\end{thebibliography}

\begin{appendices}
\section{Proof of Theorem~\ref{partition-reduction}}
\label{proof-partition-reduction}

\begin{observation}[restated Observation~\ref{partition-reduction}]
    If we can weakly $(1-\eps)$-approximate \ssum in time $\Ot(T(n,\eps))$, then we
    can $(1-\eps)$-approximate \partition in the same $\Ot(T(n,\eps))$ time.
\end{observation}

\begin{proof}
    Let $|Z| = n$ be the initial set of items.
    We run a \weak algorithm for \ssum 
    with target $b =\Sigma(Z)/2$. Let $Z^*$ denote the optimal partition
    of set $Z$:
    \begin{displaymath}
        Z^{*} = \argmax_{Z'\subseteq Z,\, \Sigma(Z') \le b}  \Sigma(Z').
    \end{displaymath}
        
    By the definition of \weak for \ssum we get a solution $Z_W$ such that:

    \begin{displaymath}
        (1-\eps)\Sigma(Z^*) \le \Sigma(Z_W) \;\;\; \text{and}  \;\;\; \Sigma(Z_W) < (1+\eps)b
    \end{displaymath}

    If $\Sigma(Z_W) \le b$ then it is a correct solution for
    \partition.
    Otherwise we take a set $Z'_W = Z \setminus Z_W$. Because
    $\Sigma(Z)/2 = b$ we know that $\Sigma(Z'_W) < b$. Additionally we
    know, that $\Sigma(Z_W) < (1+\eps)b$, so $(1-\eps)b < \Sigma(Z'_W)$.
    Similarly, because $Z^* \le b$, we have:
    \begin{displaymath}
        (1-\eps)\Sigma(Z^*) \le (1-\eps)b < \Sigma(Z'_W) \le \Sigma(Z^*) \le b.
    \end{displaymath}

    So $\Sigma(Z'_W)$ follows the definition of approximation for \partition.
    The running time follows because $T(n,1/\eps)$ must be superlinear
    (algorithm needs to read input at least) and we executed the \weak \ssum
    algorithm only constant number of times.
\end{proof}

%
%
%

\section{Proofs of the Rounding Lemmas}
\label{proof-numbers-rounding}

\newcommand{\sumx}{\sum_{i=1}^k x_i}

\begin{lemma}
    \label{knumbers-rounding}
    For $k$ natural numbers $x_1,x_2,\ldots, x_k$ and positive $q,\eps$ such that
    $q \le \sumx$ and $0 < \eps < 1$, it holds:

    \begin{displaymath}
        \sumx \le \frac{q\eps}{k} \sum_{i=1}^k \ceil{\frac{k x_i}{q \eps} } <
        (1+\eps)\sumx
    \end{displaymath}
\end{lemma}
    
\begin{proof}
    Let $x_i = \frac{q\eps}{k} c_i + d_i$ where $0 < d_i \le \frac{q\eps}{k}$.
    If some $x_i = 0$ then we set $c_i = d_i = 0$, however we know there is at least
    one positive $d_i$ (we will use this fact later).
    We have $\ceil{ \frac{k x_i}{q\eps} } = c_i+1$.
First, note that:
\begin{displaymath}
\sumx = \frac{q\eps}{k} \sum_{i=1}^k c_i +
\sum_{i=1}^k d_i \le \frac{q\eps}{k} \sum_{i=1}^k c_i + q\eps = \frac{q\eps}{k}\sum_{i=1}^k( c_i +1),
\end{displaymath}
what proves the left inequality. To handle the right inequality we take advantage of the assumption $\sumx \ge q$ and get:

\begin{align*}
    (1+\epsilon)\sumx = \eps\sumx + \sumx \ge q\epsilon + \sumx 
    = q\epsilon + \frac{q\eps}{k}\sum_{i=1}^k c_i + \sum_{i=1}^k d_i =\\
    \frac{q\eps}{k}\sum_{i=1}^k (c_i + 1) + \sum_{i=1}^k d_i
    > \frac{q\eps}{k} \sum_{i=1}^k (c_i + 1).
\end{align*}
\end{proof}

\begin{lemma}
\label{rounding-numbers-d}
For $k$ natural numbers $x_1,x_2,\ldots, x_k$ and positive $q,\eps$ such that
    $q \le \sumx$ and $0 < \eps < 1$, it holds:

\begin{displaymath}
    (1-\eps)\sumx <
        \frac{q\eps}{k} \sum_{i=1}^k\floor{\frac{k x_i}{q\eps}}
    \le \sumx.
\end{displaymath}

\end{lemma}

\begin{proof}
    The proof is very similar to the proof of
    Lemma~\ref{knumbers-rounding}, however now we represent
    $x_i$ as $\frac{q\eps}{k} c_i + d_i$ 
    where $0 \le d_i < \frac{q \eps}{k}$. We have
    $\floor{\frac{k x_i}{q \eps}} = c_i$.
    The right inequality holds because:
    
    \begin{displaymath}
        \sumx = \frac{q \eps}{k} \sum_{i=1}^k c_i  + \sum_{i=1}^k d_i \ge
        \frac{q \eps}{k} \sum_{i=1}^k c_i
    \end{displaymath}
    and the left one can be proven as follows:
    \begin{eqnarray*}
        (1-\eps)\sumx = \sumx - \eps\sumx \le \sumx - q\eps \ =  \left(\sum_{i=1}^k \frac{q \eps}{k} c_i + d_i\right) - q\eps < \frac{q\eps}{k} \sum_{i=1}^k c_i. 
    \end{eqnarray*}
\end{proof}

\section{Problems Definitions}
\label{problems-definitions}

\subsection{Exact problems}
\defproblem{\knapsack}
{A set of $n$ items $\{(v_1,w_1),\ldots, (v_n,w_n)\}$ }
{Find $x_1,\ldots,x_n$ such that:
    \begin{equation*}
        \begin{array}{ll@{}ll}
            \text{maximize}  & \displaystyle\sum\limits_{j=1}^{n} v_{j}&x_{j}
            &\\
            \text{subject to}& \displaystyle\sum\limits_{j=1}^n &w_{j} x_{j} \leq t,  &\\
            &
            &x_{j} \in \{0,1\}^n, &j=1 ,\ldots n
            .
        \end{array}
    \end{equation*}
}

Sometimes, instead of exact solution $x_1,\ldots,x_n$ in \knapsack-type problems
one needs to return the value of such solution. In decision version of such
problems we are given capacity $t$ and value $v$ and ask if there is a subset of
items with the total capacity not exceeding $t$ and total value exactly $v$ (e.g., see discussion
in~\cite{icalp2017,knapsack-book}).

\defproblem{\uknapsack}
{A set of $n$ items $\{(v_1,w_1),\ldots, (v_n,w_n)\}$ }
{Find $x_1,\ldots,x_n$ such that:
    \begin{equation*}
        \begin{array}{ll@{}ll}
            \text{maximize}  & \displaystyle\sum\limits_{j=1}^{n} v_{j}&x_{j}
            &\\
            \text{subject to}& \displaystyle\sum\limits_{j=1}^n &w_{j} x_{j} \leq t,  &\\
            &
            &x_{j} \in \mathbb{N} \cup \{0\}, &j=1 ,\ldots n
            .
        \end{array}
    \end{equation*}
}

\defproblem{\ssum}
{A set of $n$ integers $\{w_1,\ldots,w_n\}$ }
{Find $x_1,\ldots,x_n$ such that:
    \begin{equation*}
        \begin{array}{ll@{}ll}
            \text{maximize}  & \displaystyle\sum\limits_{j=1}^{n} w_{j}&x_{j}
            &\\
            \text{subject to}& \displaystyle\sum\limits_{j=1}^n &w_{j} x_{j} \leq t,  &\\
            &
            &x_{j} \in \{0,1\}^n, &j=1 ,\ldots n
            .
        \end{array}
    \end{equation*}
}

\defproblem{\partition}
{A set of $n$ integers $\{w_1,\ldots,w_n\}$ and $b = \frac{1}{2}\sum_{i=1}^n w_i$}
{Find $x_1,\ldots,x_n$ such that:
    \begin{equation*}
        \begin{array}{ll@{}ll}
            \text{maximize}  & \displaystyle\sum\limits_{j=1}^{n} w_{j}&x_{j}
            &\\
            \text{subject to}& \displaystyle\sum\limits_{j=1}^n &w_{j} x_{j} \leq b,  &\\
            &
            &x_{j} \in \{0,1\}^n, &j=1 ,\ldots n
            .
        \end{array}
    \end{equation*}
}

\defproblem{\minconv}
{Sequences $(a[i])_{i=0}^{n-1},\, (b[i])_{i=0}^{n-1}$}
{Output sequence $(c[i])_{i=0}^{n-1}$, such that $c[k] = \min_{i+j=k} (a[i]+b[j])$}

\defproblem{\ksum}
{$k-1$ sets $A_1,A_2,\ldots,A_{k-1}$ and the set $S$ of integers, each with cardinality at most $n$.}
{Is there a $(a_1,\ldots,a_{k-1},s) \in A_1\times \ldots\times A_{k-1} \times S$ such that $a_1+\ldots +a_{k-1}=s$}

\defproblem{\IIIsum}
{3 sets $A,B,C$ of integers, each with cardinality at most $n$.}
{Is there a triple $(a,b,c) \in A\times B \times C$ such that $a+b=c$}

\defproblem{\sparsity}
{A rooted tree $T$ with a weight function $x:V(T) \rightarrow \mathbb{N}$, parameter $k$}
{Find the maximal total weight of a rooted subtree of size $k$}

\subsection{Approximate problems definition}
Let $\Sigma(S)$ denote the sum of elements in $S$. The $V(I)$ denotes the total
value of items $I$ and $W(I)$ denotes the total weight of items.

\defproblem{$(1-\eps)$-approximation of \knapsack}
{A set $S = \{(v_1,w_1),\ldots, (v_n,w_n)\}$ items and a target number $t$}
{Let $Z^*$ be the optimal solution of exact \knapsack with target $t$. The
$(1-\eps)$-approximate algorithm for \knapsack returns $Z_H$ such that
$(1-\eps)V(Z^*) \le V(Z_H) \le V(Z^*)$ and $W(Z_H) \le t$.
}

Analogous definition is for \uknapsack.

\defproblem{$(1-\eps)$-approximation of \ssum}
{A set $S = \{a_1,\ldots, a_n\}$ of positive integers and a target number $t$}
{Let $Z^*$ be the optimal solution of exact \ssum with target $t$. The
$(1-\eps)$-approximate algorithm returns $Z_H$ such that $(1-\eps)\Sigma(Z^*) \le \Sigma(Z_H) \le \Sigma(Z^*)$
}

\defproblem{$(1-\eps)$-approximation of \partition}
{A set $S = \{a_1,\ldots, a_n\}$ of positive integers}
{Let $Z^*$ be the optimal solution of exact \partition. The
$(1-\eps)$-approximate algorithm returns $Z_H$ such that $(1-\eps)\Sigma(Z^*) \le \Sigma(Z_H) \le \Sigma(Z^*)$
}

\defproblem{Weak $(1-\eps)$-approximation of \ssum}
{A set $S = \{a_1,\ldots, a_n\}$ of positive integers and a target number $t$}
{Let $Z^*$ be the optimal solution of exact \ssum with target $t$. The
$(1-\eps)$-approximate algorithm returns $Z_H$ such that $(1-\eps)\Sigma(Z^*)
\le \Sigma(Z_H) \le \Sigma(Z^*)$ or $t \le \Sigma(Z_H) \le (1+\eps)t$ }

\defproblem{Approximate \IIIsum (\cite{3sum-approx})}
{Three sets $A$, $B$, $C$ of positive integers, each with cardinality at most $n$.}
{The algorithm:
\begin{itemize}
    \item concludes that no triple $(a,b,c) \in A\times B\times C$ with $a+b=c$
        exists, or
    \item it outputs a triple $(a,b,c) \in A\times B \times C$ with $a+b \in [c/(1+\eps),c(1+\eps)]$
\end{itemize}}

\defproblem{Approximate \minconv}
{Sequences \sequence{A}{0}{n-1}, \sequence{B}{0}{n-1} of positive integers and approximation parameter $0 < \eps \le 1$}
{Assume that $\text{OPT}[k] = \min_{0 \le i \le k} (A[i]+B[k-i])$ is the exact \minconv of
$A$ and $B$. The task is to output a sequence \sequence{C}{0}{n-1} such that
$\forall_i\, \text{OPT}[i] \le C[i] \le (1+\eps)\text{OPT}[i]$}
\end{appendices}

\end{document}